\documentclass[sigconf,preprint]{acmart}

\usepackage{booktabs} 

\setcopyright{none}

\acmDOI{}

\acmISBN{}

\acmConference[]{}{}{}

\usepackage{booktabs}   
\usepackage{subcaption} 

\usepackage{listings}

\usepackage{listings-golang} 

\usepackage{bchart}

\makeatletter
\newcommand{\xRightarrow}[2][]{\ext@arrow 0359\Rightarrowfill@{#1}{#2}}
\makeatother

\newcommand{\bi}{\begin{array}[t]{@{}l@{}}}
\newcommand{\ei}{\end{array}}
\newcommand{\ba}{\begin{array}}
\newcommand{\ea}{\end{array}}
\newcommand{\bda}{\[\ba}
\newcommand{\eda}{\ea\]}
\newcommand{\bp}{\begin{quote}\tt\begin{tabbing}}
\newcommand{\ep}{\end{tabbing}\end{quote}}

\newcommand{\ignore}[1]{}

\newcommand{\ms}[1]{{\bf MS: #1}}

\newcommand{\mathem}{\sf}

\newcommand{\figurebox}[1]
        {\fbox{\begin{minipage}{\textwidth} #1 \end{minipage}}}
\newcommand{\boxfig}[3]
        {\begin{figure*}\figurebox{#3\caption{\label{#1}#2}}\end{figure*}}        

\def\ruleform#1{{\setlength{\fboxrule}{1pt}\fbox{\normalsize $#1$}}}

\newcommand{\myirule}[2]{{\renewcommand{\arraystretch}{1.2}\ba{c} #1
                      \\ \hline #2 \ea}}

\newcommand{\rlabel}[1]{\mbox{(#1)}}

\newcommand{\gtTr}[2]{#1 <_{tr} #2}

\newcommand{\thread}[2]{#1 \sharp #2}
\newcommand{\link}[2]{#1 | #2}
\newcommand{\clock}[2]{#1^{#2}}
\newcommand{\ppclock}[3]{^{#2}#1^{#3}}
\newcommand{\ploc}[2]{(#1)_{#2}}

\newcommand{\absent}{\bullet}

\newcommand{\varState}[3]{#1^{(#2,#3)}}

\newcommand{\initVC}{\overline{0}}

\newcommand{\incC}[2]{{\mathem inc}(#1,#2)}
\newcommand{\maxC}[2]{{\mathem max}(#1,#2)}

\newcommand{\dom}[1]{{\mathit dom}(#1)}

\newcommand{\sep}[2]{(#1 \mid #2)}

\newcommand{\close}[1]{\mathit{close}(#1)}
\newcommand{\select}{\mathit{select}}

\newcommand{\first}{{\mathem fst}}
\newcommand{\second}{{\mathem snd}}

\newcommand{\pp}{\ \texttt{++}}

\newcommand{\GO}{\mbox{\mathem spawn}}
\newcommand{\SELECT}{\mbox{\mathem select}}

\newcommand{\MAKE}{\mbox{\mathem make}}
\newcommand{\MAKECHAN}[1]{\mbox{\mathem make}({\mathem chan} \ #1)}
\newcommand{\SYNCMAKECHAN}{\MAKECHAN{0}} 
\newcommand{\SEND}[2]{#1 \leftarrow #2}
\newcommand{\RCV}[1]{\leftarrow #1}
\newcommand{\CLOSE}[1]{\mbox{\mathem close}(#1)}

\newcommand{\tid}{\mbox{\mathem{tid}}}
\newcommand{\pc}{\mbox{\mathem{pc}}}

\newcommand{\atomicInc}[1]{\mathem{atomicInc}(#1)}

\newcommand{\pre}[1]{\mathit{pre}(#1)}
\newcommand{\post}[1]{\mathit{post}(#1)}
\newcommand{\postSnd}[3]{\mathit{post}(#1,#2, #3 !)} 
\newcommand{\postRcv}[3]{\mathit{post}(#1,#2, #3 ?)} 
\newcommand{\postRead}[1]{\mathit{read}(#1)}
\newcommand{\postWrite}[1]{\mathit{write}(#1)}

\newcommand{\postClose}[1]{\mathit{close}(#1)}
\newcommand{\postDefault}{\mathit{default}}
\newcommand{\Default}{\mathit{default}}

\newcommand{\signalTrace}[1]{\mathit{signal}(#1)}
\newcommand{\waitTrace}[1]{\mathit{wait}(#1)}

\newcommand{\instrt}[1]{\mathit{instr}(#1)}

\newcommand{\retrieve}[1]{\mathit{retr}(#1)}

\newcommand{\semP}[3]{#1 \xRightarrow{#2} #3}

\newcommand{\replay}[3]{\semP{#1}{#2}{#3}}

\newcommand{\assign}{:=}

\newcommand{\sndEvt}[2]{#1 \sharp #2!}
\newcommand{\rcvEvt}[2]{#1 \sharp #2?}
\newcommand{\readEvt}[2]{#1 \sharp \mathit{read}(#2)}
\newcommand{\writeEvt}[2]{#1 \sharp \mathit{write}(#2)}
\newcommand{\closeEvt}[2]{#1 \sharp \mathit{close}(#2)}
\newcommand{\defaultEvt}[1]{#1 \sharp \mathit{default}}
\newcommand{\initEvt}[2]{#1 \sharp \mathit{init}(#2)}

\newcommand{\snd}[1]{#1 !}
\newcommand{\rcv}[1]{#1 ?}

\newcommand{\dummyTID}{\infty}

\makeatletter
\newdimen\legendxshift
\newdimen\legendyshift
\newcount\legendlines
\newcommand{\bclldist}{1mm}
\newcommand{\bclegend}[3][10mm]{%
	\legendxshift=0pt\relax
	\legendyshift=0pt\relax
	\xdef\legendnodes{}%
	\foreach \lcolor/\ltext [count=\ll from 1] in {#3}%
	{\global\legendlines\ll\pgftext{\setbox0\hbox{\bcfontstyle\ltext}\ifdim\wd0>\legendxshift\global\legendxshift\wd0\fi}}%
	\@tempdima#1\@tempdima0.5\@tempdima
	\pgftext{\bcfontstyle\global\legendxshift\dimexpr\bcwidth-\legendxshift-\bclldist-\@tempdima-0.72em}
	\legendyshift\dimexpr5mm+#2\relax
	\legendyshift\legendlines\legendyshift
	\global\legendyshift\dimexpr\bcpos-2.5mm+\bclldist+\legendyshift
	\begin{scope}[shift={(\legendxshift,\legendyshift)}]
		\coordinate (lp) at (0,0);
		\foreach \lcolor/\ltext [count=\ll from 1] in {#3}%
		{
			\node[anchor=north, minimum width=#1, minimum height=5mm,fill=\lcolor] (lb\ll) at (lp) {};
			\node[anchor=west] (l\ll) at (lb\ll.east) {\bcfontstyle\ltext};
			\coordinate (lp) at ($(lp)-(0,5mm+#2)$);
			\xdef\legendnodes{\legendnodes (lb\ll)(l\ll)}
		}
		\node[draw, inner sep=\bclldist,fit=\legendnodes] (frame) {};
	\end{scope}
}
\makeatother

\begin{document}

\title{Two-Phase Dynamic Analysis of Message-Passing Go Programs based on Vector Clocks}     


\author{Martin Sulzmann}

\affiliation{%
  \institution{Karlsruhe University of Applied Sciences}
  \streetaddress{Moltkestrasse 30}
  \city{Karlsruhe, Germany}
  \postcode{76133}
}
  \email{martin.sulzmann@gmail.com} 

\author{Kai Stadtm{\"u}ller}
   
\affiliation{%
  \institution{Karlsruhe University of Applied Sciences}
  \streetaddress{Moltkestrasse 30}
  \city{Karlsruhe, Germany}
  \postcode{76133}
}

   \email{kai.stadtmueller@live.de} 

\begin{abstract}
  Understanding the run-time behavior of concurrent programs is a challenging task.
  A popular approach is to establish a happens-before relation via vector clocks.
  Thus, we can identify bugs and performance bottlenecks, for example,
   by checking if two conflicting events may happen concurrently.
   We employ a two-phase method to derive vector clock information for a wide range of concurrency features that includes all of the message-passing features in Go.
   The first phase (instrumentation and tracing) yields a run-time trace that records all events
   related to message-passing concurrency that  took place.
  The second phase (trace replay) is carried out offline and replays
  the recorded traces to infer vector clock information.
  Trace replay operates on thread-local traces.
  Thus, we can observe behavior that might result from
  some alternative schedule.
  Our approach is not tied to any specific language.
  We have built a prototype for the Go programming language
  and provide empirical evidence of the usefulness of our method.
\end{abstract}



\maketitle

\section{Introduction}

The analysis of concurrent programs is an important but due to the high degree
of non-determinism a notoriously difficult problem.
We consider here programs that make use of message-passing in
the style of Communicating Sequential Processes (CSP)~\cite{Hoare:1978:CSP:359576.359585}.
In our implementation, we support the Go programming language~\cite{golang}
but our approach also applies to languages with similar message-passing
features such as Concurrent ML (CML)~\cite{Reppy:1999:CPM:317040}.
Our focus is on the dynamic analysis (a.k.a.~testing) of Go where we assume
that the program is executed for a fixed number of steps.
Specifically, we consider the challenge of given a precise explanation
of the interplay among message-passing events that take place for a single execution run.

Consider the following program where we adopt Go-style notation for message-passing.
\begin{verbatim}
  spawn { x <- 1 };    // M1
  spawn { <-x };       // M2
  <-x;                 // M3
\end{verbatim}
We assume that \texttt{x} is some unbuffered channel. We write \texttt{x <- 1} to send value \texttt{1}
via channel \texttt{x} and write \texttt{<-x} to receive some value via \texttt{x}.
The actual values sent/received do not matter here.

We consider a program run where location M3 receives a value from M1.
The receive at location M2 is blocked but from the user perspective the program terminates without showing
any abnormal behavior.
In Go, once the main thread terminates, all remaining threads are terminated as well.
Our analysis is able to feedback to the user that M2 could have also received a value from M1
(which then would result in a deadlock).
Suppose we encounter the deadlock. That is, M2 receives from M1.
For such deadlocking situation, the Go run-time reports  for each blocked thread
the event that is responsible for the blockage.
Our analysis provides more details and 
feedbacks to the user that M3 is blocked (the main thread) but
could possibly communicate via M1.

In the following example, we make use of Go's ability to close a channel.
Any subsequent receive on a closed channel never blocks and obtains a default value.
However, any subsequent send fails and yields a run-time exception.
\begin{verbatim}
  spawn { x <- 1 };       // M1
  spawn { close(x) };     // M2
  <-x;                    // M3
\end{verbatim}
Assuming thread M2 executes followed by M3. From the user perspective,
the program terminates without any abnormal behavior.
Recall that once the main thread terminates, all remaining threads such as M1 terminate as well.
Our analysis feedbacks to the user that there might be a different schedule
where a send on a closed channel may occur.

In our final example, we consider Go's \texttt{select} statement 
which corresponds to non-deterministic choice.
\begin{verbatim}  
  spawn { x <- 1;           // M1
          y <- 1 };         // M2
  select {                   
    case <-x:               // M3
    case <-y:               // M4
  } 
\end{verbatim}
The \texttt{select} statement blocks if neither of the cases is available, i.e.~can communicate with some concurrent event.
If both cases are available, one of the cases is chosen based on a pseudo-random order to ensure fairness.
For our example, case M4 never applies. The user may wonder why this is so.
Our analysis feedbacks to the user that (for this specific execution run), any potential communication partner of M4
always happens after a communication of M3 took place.

Besides assisting the user in narrowing down the source of a bug, our analysis can also identify
performance bottlenecks. For example, consider the case of a (too slow) receiving thread that needs to negotiate
with a high number of sending parties. A specific instance of this analysis case can be used for lock contention.
We do not support the concept of a mutex in our formal treatment. However, it is well known that a mutex
effectively corresponds to a buffered channel of size one where
send equals to lock and receive equals to unlock. 

To provide the above feedback to the user,
we infer the dependencies among concurrent events.
This is achieved by establishing a happens-before relation~\cite{lamport1978time}
where the happens-before relation is derived via vector clocks~\cite{Fidge:1991:PAT:646210.683620,Mattern89virtualtime}.
Earlier work by Fidge~\cite{fidge1988timestamps} and Mattern~\cite{Mattern89virtualtime}
shows how to compute vector clocks in the message-passing setting.
We improve on these results as follows.
First, we also cover buffered channels of a fixed size.
Second, we introduce a novel form of a pre vector clock annotation to analyze
events that lack a communication partner. That is, events that could not commit such as M2 in our first example above.

A novel aspect of our work is a two-phase method to derive vector clocks.
The first phase consists of a light-weight instrumentation of the program that 
can be carried out via a simple pre-processing step.
The inference of vector clock information happens in a subsequent (off-line) phase.
Events are not recorded in a global trace as they actually took place.
Rather, we record events on a per-thread basis.
Hence, tracing requires little synchronization and therefore incurs only a low run-time overhead.

Almost no extra synchronization for tracing purposes among threads is required.
To guarantee that we derive vector clock information that corresponds to an actual program run,
a receive event obtains the thread id and program counter from the sending party.
This is the only extra intra-thread information required to properly match a receive event
to its corresponding send event.

In summary, our contributions are:
\begin{itemize}
\item Based on a simple instrumentation method to obtain thread-local run-time traces of recorded events (Section~\ref{sec:instrumentation}),
      we give a precise account of how to derive vector clock information for a wide
      range of message-passing features (Section~\ref{sec:vcs}).

\item We discuss several analysis scenarios where the vector clock information inferred
  proves to be useful (Section~\ref{sec:analysis}).
  The scenarios comprise detection of performance bottlenecks and potential bugs
    as well as recovery from a bug.

\item  We have built a prototype for the Go programming language~\cite{golang}
       and provide experiments that include real-world examples where we
       discuss the effectiveness of our method (Section~\ref{sec:experiments}).            
      
\end{itemize}

Related work is discussed in Section~\ref{sec:related-work}.
We conclude in Section~\ref{sec:conclusion}.
Further details can be found in the Appendix.

\section{Instrumentation and Tracing}
\label{sec:instrumentation}

  \boxfig{f:programs}{Programs}{    
  \bda{lcll}
  x,y, &\dots & &
  \mbox{Variables, Channel Names}
  \\ i,j, &\dots && \mbox{Integers}
  \\
  b & ::= &  x \mid i \mid \first(b) \mid \second(b) \mid (b,b)   & \mbox{Expressions}
  \\
  e,f & ::= & \SEND{x}{b} \mid y \assign \RCV{x}  & \mbox{Receive/Transmit}
  \\
  c & ::= & y \assign z \mid y \assign \MAKECHAN{i} \mid \CLOSE{x} \mid \GO\ p  & \mbox{Commands}
  \\ & \mid & \SELECT\ [e_i \Rightarrow p_i]_{i\in I} \mid \SELECT\ [e_i \Rightarrow p_i \mid p]_{i\in I}
  \\
  p,q,r & ::= & [] \mid c : p   & \mbox{Program}
  \eda
  }

We assume a simplified language to cover the main concurrency features of Go. See Figure~\ref{f:programs}.
Expressions include pairs (anonymous struct).
Notation for send and receive follows Go syntax.
Receive is always tied to an assignment statement
where we write $\assign$ to denote assignment. Type declarations of variables are omitted for brevity.
A buffered channel $x$ of size $i>0$ is introduced via $x \assign \MAKECHAN{i}$.
For $i=0$, we refer to $x$ as an unbuffered channel.
Commands to spawn a new thread and close a channel we have seen already.

Go supports non-deterministic choice via \SELECT\ where the to be selected cases
are represented in a list.
For example,
\bda{cc}
\SELECT\ [ \SEND{x}{b} \Rightarrow \dots, y \assign \RCV{z} \Rightarrow \dots ] & \rlabel{SEL}
\eda
denotes a command that either sends a value via channel $x$
or receive a value from channel $z$.
We also support \SELECT\ with a default case.
We assume that a single send/receive statement is represented by a \SELECT\ statement
with a single case.
We ignore locks as their treatment exactly corresponds 
to buffered channels of size one.

A program is represented as a list of commands.
  We follow Haskell style syntax
  and write $c : p$ to denote a non-empty list with head $c$ and tail $p$.
  We write $\pp$ to denote list concatenation.

  Programs are instrumented to record the events
  that took place when executing the program.
  Events are recorded on a per thread basis.
  Hence, we obtain a list of (thread-local) traces where each trace is connected to a thread.
We write $[\thread{1}{T_1}, \dots, \thread{n}{T_n}]$
to denote the list of recorded traces attached with their thread id.
The syntax of traces and events we use is as follows.

\begin{definition}[Run-Time Traces and Events]
\label{def:run-time-traces-events}  

  \bda{lcll}
  U,V & ::= & [] \mid \thread{i}{T} : U  & \mbox{Thread-local traces}
  \\ T & ::= & [] \mid t : T & \mbox{Trace}
  \\ t & ::= & \signalTrace{i} \mid \waitTrace{i} & \mbox{Events}
   \\ & \mid & \pre{as} \mid \postSnd{i}{i}{x}
                 \mid \postRcv{i}{i}{x}
                 \\ & \mid & \postClose{x} \mid \post{\Default} 
  \\
       a,b & ::= & \snd{x} \mid \rcv{x} \mid \Default          
    \\ as & ::= & [] \mid  a : as 
  \eda
\end{definition}

The purpose of each event becomes clear when considering
the instrumentation of programs.

  \boxfig{f:instrumentation}{Instrumentation}{      
   \bda{lcl}
 \instrt{[]} & = & []
  \\
 \instrt{c:p} & = & \instrt{c} : \instrt{p}
 \\
 \\
  \instrt{y \assign z} & = & [y \assign z] 
  \\
  \instrt{\GO\ p} & = & [i \assign \atomicInc{\mathit{cnt}}, x_{tid} \assign x_{tid} \pp\ [\signalTrace{i}],
    \\ & &             \GO\ ([x_{tid} \assign [\waitTrace{i}]] \pp\ \instrt{p})]
  \\
  \\
  \instrt{\CLOSE{x}} & = & [\CLOSE{x}] \pp\ [x_{tid} \assign x_{tid} \pp [\postClose{x}]]          
  \\
  \\
  \instrt{y \assign \MAKECHAN{i}} & = & [y \assign \MAKECHAN{i}]              
  \\
  \\
  \instrt{\SELECT\ [e_i \Rightarrow p_i]_{i\in \{1,\dots,n\}}}
  & =  & [x_{tid} \assign x_{tid} \pp\ [\pre{[\retrieve{e_1},\dots,\retrieve{e_n}]}],
    \\ & &\SELECT\ [\instrt{e_i \Rightarrow p_i}]_{i \in \{1,\dots,n\}}]
  \\
  \\
   \instrt{\SELECT\ [e_i \Rightarrow p_i | p]_{i\in \{1,\dots,n\}}}
   & =  & [x_{tid} \assign x_{tid} \pp\ [\pre{[\retrieve{e_1},\dots,\retrieve{e_n},\Default]}],
          \\ & &\SELECT\ [\instrt{e_i \Rightarrow p_i} | [ x_{tid} \assign x_{tid} \pp [\post{\Default}]] \pp\ \instrt{p} ]_{i \in \{1,\dots,n\}}]
   \\
   \\
  \instrt{\SEND{x}{b} \Rightarrow p}
      & =  & \SEND{x}{((\tid,\pc),b)} \Rightarrow
                  (x_{tid} \assign x_{tid} \pp\ [\postSnd{\tid}{\pc}{x}]) \pp\ \instrt{p}
  \\
  \\
  \instrt{y \assign \RCV{x} \Rightarrow p}
   & = & y' \assign \RCV{x} \Rightarrow
          [x_{tid} \assign x_{tid} \pp\ [\postRcv{\first(\first(y'))}{\second(\first(y'))}{x}]],           
            \\ &&   \ \ \ \ \ \ \ \ \ \ \ \ \ \ \ \ \ \ \ \ \ \ \  y \assign \second(y')]   \pp\ \instrt{p}
  \eda
  
    \bda{c}
  \retrieve{\SEND{x}{b}}  =  \snd{x}
  \ \ \ \  \retrieve{y = \RCV{x}}  =  \rcv{x}
  \eda
  }

  Figure~\ref{f:instrumentation} formalizes the instrumentation of programs
  which can be carried out via a simple pre-processor.
  For each thread, we assume a thread-local variable $x_{tid}$ that stores the events
  that take place in this thread.
  In turn, we discuss the various instrumentation cases.

  As we support dynamic thread creation,
  there might be dependencies among threads when it comes to tracing.
For example, consider the following program snippet.
\bda{l}
[ \dots , \GO\ [\dots]].
\eda
Our instrumentation yields
\bda{l}
    [ \dots, x_1 \assign x_1 \pp  [\signalTrace{1}], \GO\ [x_2 \assign [\waitTrace{1}], \dots]].
\eda
Variable $x_1$ logs the events of the main thread, and variable $x_2$ the events
in the newly spawned thread. Events in $x_2$ are only processed
once the wait event is matched against its corresponding signal event.
Thus, we ensure that events logged in a newly created thread happen after
the events that took place in the thread that issued the $\GO$ command.
In the instrumentation, we assume a shared variable $\mathit{cnt}$
where the primitive $\atomicInc{\mathit{cnt}}$ atomically increments this variable
and returns the updated value.

In Go it is possible to close a channel which means that any subsequent send
on that channel yields an error but any receive succeeds by retrieving a dummy value.
We support this feature by recording $\postClose{x}$.

We support unbuffered as well as buffered channels of a fixed size.
For operations on buffered channels, we also need to record the buffer size.
This can be easily done in the instrumentation but is left out for brevity.

In case of channel operations send and receive, we use (post) events
$\postSnd{i}{i}{x}$ and $\postRcv{i}{i}{x}$
to represent committed operations. That is, sends and receives that actually took place.
To uniquely connect a sender to its corresponding receiver,
the sender transmits its thread id and program counter to the receiver.

Consider instrumentation case $\instrt{\SEND{x}{b} \Rightarrow p}$
that deals with the send operation.
We assume a primitive $\tid$ to compute the thread id
and a primitive $\pc$ to compute the thread-local program counter.
Both values are additionally transmitted to the receiver.
We assume common tuple notation.
Instead of $b$ we transmit $((\tid,\pc),b)$
and store the event $\postSnd{\tid}{\pc}{x}$.
For simplicity, we assume that both calls to $\pc$ yield the same value.
In an actual implementation, we would need to store the current program counter
and then transmit the stored value as well as record the value in the post event.
Further note that events are stored in thread-local traces.
Hence, in an implementation we could save space and drop the $\tid$ component
for committed send operations.
Here, we keep the $\tid$ component to have a uniform
presentation for send and receive.

At the receiving site, see case $\instrt{y \assign \RCV{x} \Rightarrow p}$,
we assume  primitives $\first$ and $\second$ to access the respective components
of the received value.
The receiver stores $\postRcv{\first(\first(y'))}{\second(\first(y'))}{x}$
to record the sender's thread id and program counter.

In addition to committed events we also keep track
of events that could possibly commit.
This can applies to not chosen cases in a \SELECT\ statements.
We make use of pre events to represent such cases.
For the earlier select example \rlabel{SEL},
we record both possibilities via
the event $\pre{[\snd{x}, \rcv{z}]}$.
A select statement may include a default case.
We represent this variant by including $\Default$ in the list of pre event.
If the default case is chosen, we store the event $\post{\Default}$.

The program's dynamic behavior is captured by the trace obtained
from running the instrumented program. By replaying the trace we can infer for each event a vector clock.
This is what we will discuss next.

\section{Vector Clocks}
\label{sec:vcs}

\boxfig{f:trace-replay}{Trace Replay}{

  \bda{c}
  \ruleform{\replay{\sep{Q}{U}}{E}{\sep{Q}{U}}}
  \\
  \\
  \rlabel{Shuffle} \
  \myirule{\mbox{$\pi$ permutation on $\{1,\dots,k\}$}
           \ \
           \mbox{$\rho$ permutation on $\{1,\dots,m\}$}}
          {\replay{\sep{[\thread{{x_1}^{n_1}}{B_1}, \dots, \thread{{x_k}^{n_k}}{B_k}]}
                       {[\thread{1}{T_1}, \dots, \thread{m}{T_m}]} \\}
                   {}
                   {\\ \sep{[\thread{{x_{\pi(1)}}^{n_{\pi(1)}}}{B_{\pi(1)}}, \dots, \thread{{x_{\pi(k)}}^{n_{\pi(k)}}}{B_{\pi(k)}}]}
                       {[\thread{\rho(1)}{T_{\rho(1)}}, \dots, \thread{\rho(m)}{T_{\rho(m)}}]}}}          
  \\
  \\
  \rlabel{Closure} \
  \myirule{\replay{\sep{Q_1}{U_1}}{E_1}{\sep{Q_2}{U_2}} \ \ \ \replay{\sep{Q_2}{U_2}}{E_2}{\sep{Q_3}{U_3}}}
          {\replay{\sep{Q_1}{U_1}}{E_1 \pp\ E_2}{\sep{Q_3}{U_3}}}
  \\
  \\
  \rlabel{Signal/Wait} \
\myirule{T_1 = \signalTrace{i} : T_1'
         \ \ \  T_2 = \waitTrace{i} : T_2'
        }
        {\replay{\sep{Q}{\thread{\clock{i_1}{cs_1}}{T_1} :
                 \thread{\clock{i_2}{cs_2}}{T_2} : U}}
          {[]}
            {\sep{Q}{\thread{\clock{i_1}{\incC{i_1}{cs_1}}}{T_1'} :
             \thread{\clock{i_2}{\incC{i_2}{cs_1}}}{T_2'} : U}}
        }
  \\
  \\      
\rlabel{Sync} \
\myirule{T_1 = \pre{[a_1,\dots,a_n,\snd{x}]} : \postSnd{i_1}{j}{x} : T_1'
        \\ T_2 = \pre{[b_1,\dots,b_m,\rcv{x}]} : \postRcv{i_1}{j}{x} : T_2'
        \\ cs = \maxC{\incC{i_1}{cs_1}}{\incC{i_2}{cs_2}}
        \\ E_1 = [\ppclock{\sndEvt{i_1}{x}}{cs_1}{cs}, \ppclock{\rcvEvt{i_2}{x}}{cs_2}{cs}]
        \\ E_2 = [\ppclock{i_1 \sharp a_1}{cs_1}{\absent}, \dots, \ppclock{i_1 \sharp a_n}{cs_1}{\absent},
                  \ppclock{i_2 \sharp b_1}{cs_2}{\absent}, \dots, \ppclock{i_2 \sharp b_m}{cs_2}{\absent}]
        }       
        {\replay{\sep{Q}{\thread{\clock{i_1}{cs_1}}{T_1} :
                 \thread{\clock{i_2}{cs_2}}{T_2} : U}}
          {E_1 \pp E_2}
          {\sep{Q}{\thread{\clock{i_1}{cs}}{T_1'} : \thread{\clock{i_2}{cs}}{T_2'} : U}}
        }
  \\
  \\      
  \rlabel{Send} \
  \myirule{  B = B' \pp [\bot^{cs_1},\dots,\bot^{cs_m}]
    \ \ \ \ cs'' = \maxC{\incC{i}{cs}}{cs_1}
    \\  B'' = B' \pp [\clock{\postSnd{i}{j}{x}}{cs''}, \bot^{cs_2}, \dots \bot^{cs_m}]
    \\ E = [\ppclock{\sndEvt{i}{x}}{cs}{cs''}] \pp [\ppclock{i \sharp a_1}{cs}{\absent},\dots,\ppclock{i \sharp a_n}{cs}{\absent}]
          }
          {\replay{\sep{\thread{x^k}{B} : Q}{(\thread{\clock{i}{cs}}{\pre{[a_1,\dots,a_n,\snd{x}} : \postSnd{i}{j}{x} : T} ) : U}}
                  {E}
                  {\sep{\thread{x^k}{B''} : Q}{\thread{\clock{i}{cs''}}{T} : U}}
          }
  \\
  \\
  \rlabel{Receive} \
\myirule{B = \clock{\postSnd{i_1}{j}{x}}{cs'} : B'
  \\ cs'' = \maxC{\incC{i_1}{cs}}{cs'}
  \\ E = [\ppclock{\rcvEvt{i_2}{x}}{cs}{cs''}] \pp [\ppclock{i_2 \sharp a_1}{cs}{\absent},\dots,\ppclock{i_2 \sharp a_n}{cs}{\absent}]
  }
        {\replay{\sep{\thread{x^k}{B} : Q}{\thread{\clock{i_2}{cs}}{\pre{[a_1,\dots,a_n,\rcv{x}]} : \postRcv{i_1}{j}{x} : T} : U}}
                {E}
                {\sep{\thread{x^k}{B' \pp [\bot^{cs''}]} : Q}{\thread{\clock{i_2}{cs''}}{T} : U}}
        }
  \\
  \\
 \rlabel{Receive-Closed} \
\myirule{cs' = \ \incC{i}{cs} \ \
  \\ T = \pre{[\dots,\rcv{x}, \dots]} : \postRcv{\dummyTID}{\dummyTID}{x} : T'}
        {\replay{\sep{Q}{\thread{\clock{i_2}{cs}}{T} : U}}
                {[\ppclock{\rcvEvt{i_2}{x}}{cs}{cs'}]}
                {\sep{Q}{\thread{\clock{i_2}{cs'}}{T'} : U}}
        }
 \\
 \\
  \rlabel{Close} \
 \myirule{cs' = \ \incC{i}{cs}}
         {\replay{\sep{Q}{\thread{\clock{i}{cs}}{(\postClose{x}:L)} : U}}
                 {[\clock{\closeEvt{i}{x}}{cs'}]}
                 {\sep{Q}{\thread{\clock{i}{cs'}}{L} : U}}
         }
   \\
   \\
   \rlabel{Default} \
 \myirule{cs' = \ \incC{i}{cs} \ \
   \\ T = \pre{[\dots]} : \post{\Default} : T'}
         {\replay{\sep{Q}{\thread{\clock{i}{cs}}{T} : U}}
                 {[\ppclock{\defaultEvt{i}}{cs}{cs'}]}
                 {\sep{Q}{\thread{\clock{i}{cs'}}{T'} : U}}
         }
  \eda
}


The goal is to annotate events with vector clock information.
For this purpose, we replay the set of recorded run-time traces $T_i$
to derive a global trace $E$ of vector clock annotated events.
The syntax is as follows.

\begin{definition}[Vector Clock Annotated Events]
\bda{lcll}
cs & ::= & \absent \mid [] \mid n : cs   & \mbox{Vector clock}
\\
e & ::= & \ppclock{\sndEvt{i}{x}}{cs}{cs}         & \mbox{Annotated Events}
\\  & \mid & \ppclock{\rcvEvt{i}{x}}{cs}{cs}
 \\ & \mid & \clock{\closeEvt{i}{x}}{cs}
        \mid \ppclock{\defaultEvt{i}}{cs}{cs}
\\
E & ::= & [] \mid e : E      
\eda
\end{definition}

For convenience, we represent a vector clock as a list of clocks where the first position
belongs to thread 1 etc.
We include $\absent$ to deal with events that did not commit.
More on this shortly.
We write $\initVC$ to denote the initial vector clock
where all entries (time stamps) are set to~$0$.
We write $cs[i]$ to retrieve the $i$-th component in $cs$.
We define $cs_1 > cs_2$ if for each position $i$ we have that $cs_1[i] > cs_2[i]$.
We write $\incC{i}{cs}$ to denote the vector clock obtained from $cs$
where all elements are the same but at index $i$ the element is incremented by one.
We write $\maxC{cs_1}{cs_2}$ to denote the vector clock where we per-index take
the greater element.
We write $\overline{i}$ to denote the vector clock $\incC{i}{\initVC}$, i.e.~all entries are zero
except position~$i$ which is equal to one.
We write $\clock{i}{cs}$ to denote thread $i$ with vector clock $cs$.

We write $\ppclock{\sndEvt{i}{x}}{cs_1}{cs_2}$
to denote a send operation via channel~$x$ in thread~$i$.
We infer two vector clock annotations~$cs_1$ and $cs_2$ for the following reason.
In the (run-time) trace $T$, we record for each channel operation a pre event (communication about to happen)
and a post event (communication has happened). Vector clock $cs_1$ corresponds to the pre event
and $cs_2$ to the post event. 

We write $\ppclock{\rcvEvt{i}{x}}{cs_1}{cs_2}$ to denote a
vector clock annotated receive event in thread $i$.
As in case of send, $cs_1$ represents the vector clock of the pre event
and $cs_2$ the vector clock of the post event.

We write $\ppclock{\defaultEvt{i}}{cs_1}{cs_2}$ to denote a vector clock annotated default event connected
to a select statement. Like in case of send and receive, we find pre and post vector clock annotations.
We will argue later that having
the vector clock information for the pre event can have significant advantages for the analysis.
In fact, as we support selective communication, the post vector clock for not selected
cases may be absent. We introduce the following notation.

\begin{definition}[Not Selected Events]
\label{def:not-selected-event}  
  We write $\ppclock{i \sharp a}{cs}{\absent}$ to denote an event $a$
  from thread $i$ with pre vector clock $cs$
  where the post vector clock is absent.
  For $a=x!$, $\ppclock{ \sharp a}{cs}{\absent}$ is a short-hand
  for $\ppclock{\sndEvt{i}{x}}{cs}{\absent}$.
  For $a=x?$, $\ppclock{i \sharp a}{cs}{\absent}$ is a short-hand
  for $\ppclock{\rcvEvt{i}{x}}{cs}{\absent}$.
  For $a = \postDefault$, $\ppclock{i \sharp a}{cs}{\absent}$ is a short-hand
  for $\ppclock{\defaultEvt{i}}{cs}{\absent}$.
\end{definition}

We write $\clock{\closeEvt{i}{x}}{cs}$ to denote a vector clock annotated close event on channel~$x$
in thread~$i$ where $cs$ is the post vector clock.
There is no pre vector clock as close operations never block.

Figure~\ref{f:trace-replay} defines to trace replay rules to infer vector clock information.
Replay rules effectively resemble operational rewrite rules to describe the semantics of a concurrent program.
We introduce a rewrite relation $\replay{\sep{Q}{U}}{E}{\sep{Q}{U}}$
among configurations $\sep{Q}{U}$ to derive $E$.
Component $U$ corresponds to the list of thread-local run-time traces.
Recall Definition~\ref{def:run-time-traces-events}.
Component $Q$ keeps track of the list of buffered channels and their current state.
Its definition is as follows.

\begin{definition}[Buffered Channels]
  \bda{lcl}
  P, Q & ::= & [] \mid \thread{x^n}{B} : Q
  \\
  B & ::= & [] \mid \clock{\postSnd{i}{i}{x}}{cs} : B
          \mid \clock{\bot}{cs} : B
  \eda
\end{definition}
We assume that in $\thread{x^n}{B}$, $x$ refers to the channel name
and $n$ to the buffer size. Buffer size information can be obtained during run-time tracing
but we omitted this detail in the formalization of the instrumentation.
$B$ denotes the buffer.
Initially, we assume that all buffer slots are empty and filled with $\bot^{[0,\dots,0]}$
where $[0,\dots,0]$ represents the initial vector clock.
We write $B = B' \pp [\bot^{cs_1},\dots,\bot^{cs_m}]$ to denote a buffer
where all slots in $B'$ are occupied.

We have now everything in place to discuss the trace replay rules.

\subsection{Shuffling and Collection}

In our tracing scheme, we do not impose a global order among events.
Events are stored in thread-local traces.
This allows us to explore alternative schedules
by suitably rearranging (shuffle) the list of buffered channels and thread-local traces.
In terms of the replay rules, we therefore find rule \rlabel{Shuffle}.
Via rule \rlabel{Closure} we simply combine several elementary rewriting steps.

The next set of rules assume that channels and traces are suitably shuffled
as these rules only inspect the leading buffer and the two leading traces.

\subsection{Intra Thread Dependencies}

Rule \rlabel{Signal/Wait} ensures that a thread's trace is only processed
once the events stored in that trace can actually take place.
See the earlier example in Section~\ref{sec:instrumentation}.

\subsection{Unbuffered Channels}

Rule \rlabel{Sync} processes send/receive communications via some unbuffered channel.
For convenience, we assume that primitive events in the list $\pre{as}$
can be suitably rearranged.
We check for two thread-local traces where a send and receive took place
and the send and receive are a matching pair.
That is, in the actual program run, the receiver obtained the value from this sender.
A matching pair is identified by comparing the recorded thread id and program counter of the sender.
See post events $\postSnd{i_1}{j}{x}$ and $\postRcv{i_1}{j}{x}$.

Our (re)construction of vector clocks follows the method
developed by Fidge and Mattern.
We increment the time stamps of the threads involved and
exchange vector clocks. To indicate that a synchronization between two concurrent
events took place, we build the maximum.
Our novel idea is to infer pre vector clocks.
Thus, we can detect (a) alternative communications,
and (b) events not chosen within a select statement.
Recall the notation introduced in Definition~\ref{def:not-selected-event}.
For brevity, we ignore the formal treatment of (c) orphan events.
That is, events with a singleton list of pre events that lack a post event.
We can treat such events like case (b) by including a dummy post event.

Here is an example to illustrate~(a).
\begin{example}
\label{ex:sync-chan}
Consider the program annotated with thread id numbers.
\bda{lcl}
    [ x \assign \SYNCMAKECHAN, y \assign \SYNCMAKECHAN, && (1) \\
      \GO\ [\SEND{x}{1}],  && (2) \\
      \GO\ [\RCV{x}, \SEND{x}{1}], && (3) \\
      \GO\ [\SEND{y}{1}, \RCV{x}], && (4) \\
      \GO\ [\RCV{y}] && (5)
\eda
We assume a specific program run where thread 2 synchronizes with thread 3.
Thread 4 synchronizes with thread 5 and finally
thread 3 synchronizes with thread 4.
Here is the resulting trace.
For presentation purposes, we write the initial vector clock behind each thread.
\bda{ll}
    [\thread{1}{[\signalTrace{2}, \signalTrace{3}, \signalTrace{4}, \signalTrace{5}]},
      & [1,0,0,0,0]
      \\     \thread{2}{[\waitTrace{2}, \pre{\snd{x}}, \postSnd{2}{1}{x}]},
      & [0,1,0,0,0]
      \\     \thread{3}{[\waitTrace{3}, \pre{\rcv{x}}, \postRcv{2}{1}{x}, \pre{\snd{x}}, \postSnd{3}{2}{x}]},
      & [0,0,1,0,0]
      \\     \thread{4}{[\waitTrace{4}, \pre{\snd{y}}, \postSnd{4}{1}{y}, \pre{\rcv{x}}, \postRcv{3}{2}{x}]},
      & [0,0,0,1,0]
\\     \thread{5}{[\waitTrace{5}, \pre{\rcv{y}}, \postRcv{4}{1}{y}]}
    ] & [0,0,0,0,1]
\eda
For example, in thread 3, in the second program step, the send operation on channel~$x$ could commit.
Hence, we find the event $\postSnd{3}{2}{x}$.

Trace replay proceeds as follows. We process
intra-thread dependencies via rule \rlabel{Signal/Wait}.
This leads to the following intermediate step.
\bda{ll}
    [\thread{1}{[]},
      & [5,0,0,0,0]
      \\     \thread{2}{[\pre{\snd{x}}, \postSnd{2}{1}{x}]]},
      & [1,1,0,0,0]
      \\     \thread{3}{[\pre{\rcv{x}}, \postRcv{2}{1}{x}, \pre{\snd{x}}, \postSnd{3}{2}{x}]},
      & [2,0,1,0,0]
      \\     \thread{4}{[\pre{\snd{y}}, \postSnd{4}{1}{y}, \pre{\rcv{x}}, \postRcv{3}{2}{x}]},
      & [3,0,0,1,0]
\\     \thread{5}{[\pre{\rcv{y}}, \postRcv{4}{1}{y}]}
    ] & [4,0,0,0,1]
\eda

Next, we exhaustively synchronize events and attach pre/post vector clocks.
We show the final result.
For presentation purposes, instead of $\ppclock{\sndEvt{i_1}{x}}{cs'}{cs}$, 
we write the short form $\ppclock{\snd{x}}{cs'}{cs}$.
Thread ids are written on the left.
Events annotated with pre/post vector clocks are written next
to the thread in which they arise.
We omit the main thread (1) as there are no events recorded
for this thread.

\bda{ll}
(2)     & \underline{\ppclock{\snd{x}}{[1,1,0,0,0]}{[2,2,2,0,0]}}
\\ (3)  & \ppclock{\rcv{x}}{[2,0,1,0,0]}{[2,2,2,0,0]}, \ppclock{\snd{x}}{[2,2,2,0,0]}{[4,2,3,3,2]}
\\ (4) & \ppclock{\snd{y}}{[3,0,0,1,0]}{[4,0,0,2,2]}, \underline{\ppclock{\rcv{x}}{[4,0,0,2,2]}{[4,2,3,3,2]}}
\\ (5) & \ppclock{\rcv{y}}{[4,0,0,0,1]}{[4,0,0,2,2]}
\eda
Consider the underlined events.
Both are matching events, sender and receiver over the same channel.
An alternative communication
among two matching events requires both events to be concurrent
to each other.
In terms of vector clocks, concurrent means that their vector
clocks are incomparable.

However, based on their post vector clocks it appears that
the receive on channel~$x$ in thread 4 happens after
the send in thread 2 because
$[2,2,2,0,0] < [4,2,3,3,2]$.
This shows the limitations of post vector clocks
as it is easy to see that both events
represent an alternative communication.
Thanks to pre vector clocks,
this alternative communication can be detected.
We find that events are concurrent because
their pre vector clocks are incomparable,
i.e. $[1,1,0,0,0] \not< [4,0,0,2,2]$
and $[1,1,0,0,0] \not> [4,0,0,2,2]$.
\end{example}

\subsection{Buffered Channels}

Neither Fidge nor Mattern cover buffered channels.
We could emulate buffered channels by treating each send operation as if this operation
is carried out in its own thread.
However, this leads to inaccuracies.
\begin{example}
  Consider
  \bda{lcl}
      [ x \assign \MAKECHAN{1},
        \\ \SEND{x}{1},  && (1)
        \\ \GO\ [ \SEND{x}{1} ], && (2)
        \\ \RCV{x} ] && (3)
  \eda
  Assuming an emulation of buffered channels as described above,
  our analysis would report that (2) and (3) form an alternative match.
 However, in the Go semantics, buffered messages are queued.
 Hence, for \emph{every} program run the only possibility
 is that (1) synchronizes with (2) and a synchronization with (3) never takes place!
 \end{example}
 
 We can eliminate such false positives by keeping track of (un)occupied buffer space
 during trace replay.
 Rule \rlabel{Receive} processes a receive over some buffered channel.
 We check the first occupied buffer slot where buffers are treated like queues.
 The  enqueued send event must match the receive event.
 We check for a match by comparing received thread id and program counter.
 The receive events pre/post vector clocks are computed as in case of rule \rlabel{Sync}.
 The buffered send event is dequeued and we enqueue an empty buffer slot attached with the sender's vector clock.
 This is important to establish the proper order among receivers and senders as we will see shortly.

 Consider rule \rlabel{Send} where a sender synchronizes with an empty buffer slot.
 Recall that notation $B = B' \pp [\bot^{cs_1},\dots,\bot^{cs_m}]$ implies that all
 buffer slots in $B'$ are occupied.
 We increment the time stamp of the thread and synchronize with the empty buffer slot 
 by building the maximum. The now occupied buffer slot carries the resulting vector clock.
 If we would simply overwrite the buffer slot with the sender's vector clock,
 the proper order among receive and send events may get lost.

 \begin{example}
Consider 
\bda{lcl}
    [x \assign \MAKECHAN{2},
      \\ \GO\ [ \RCV{x} ],     && (1)
      \\ \SEND{x}{1},
      \\ \SEND{x}{1},
            \\ \SEND{x}{1} ]   && (2)

\eda
It is clear that for any program run the receiver at location (1) happens before the send at location (2).
Suppose we encounter a program run where the  first two sends take place before the receive.
For brevity, we omit the set of local traces containing all recorded pre/post events.
Here is the program annotated with (post) vector clock information.
\bda{lcl}
    [x \assign \MAKECHAN{2},
      \\ \GO\ [ \RCV{x} ],     && [3,2]
      \\ \SEND{x}{1},          && [3,0]
      \\ \SEND{x}{1},          && [4,0]
            \\ \SEND{x}{1} ]   && ??

\eda
Recall that the main thread (with id number~$1$) creates a new thread (signal/wait events).
This then leads to the first send having the post vector clock $[3,0]$
and the second send having the post vector clock $[4,0]$.
At this point, the buffer contains
$$
[\postSnd{1}{\_}{x}^{[3,0]}, \postSnd{1}{\_}{x}^{[4,0]}].
$$
We write $\_$ to indicate that the program counter of the sending thread~$1$ does not matter here.
Then, the receive synchronizes with the first send.
Hence, we find the post vector clock $[3,2]$ and the buffer
has the form $[\post{\thread{1}{\snd{x}}}^{[4,0]}, \bot^{[3,2]}]$.
As there is an empty buffer slot. The third send can proceed.

Ignoring the vector clock attached to the empty buffer slot would result in the (post) vector clock $[5,0]$
for the third send. This is clearly wrong as then the receive and (third) send appear to be concurrent to each other.
Instead, the sender synchronizes with the vector clock of the empty buffer slot. See rule \rlabel{Send}.
In essence, this vector clock corresponds to the vector clock of the earlier receive.
Hence, we find that the third send has the vector clock $[5,2]$
and thus the receive happens before the third send.
\end{example}

 \subsection{Closed Channel}

 We deal with close events by simply incrementing the thread's timestamp.
 See rule \rlabel{Close}.
 A receive event on a closed channel is distinguished from other receives by
 the fact that dummy values are received. We write $\dummyTID$ to refer to a dummy thread id
 and program counter.
 See rule \rlabel{Receive-Closed}.

\subsection{Select with Default} 
 
Rule \rlabel{Default} covers that case that a default branch of a select statement
has been taken.

\subsection{Properties}
\label{sec:properties}

Senders and receivers are uniquely connected based on the sender's thread id and program counter.
Hence, any vector clock annotation obtained via trace replay
corresponds to a \emph{valid} program run.
However, due to our thread-local tracing scheme,
trace replay rules do not need to follow the schedule of the \emph{actual} program run.
It is possible to explore alternative schedules.
In case of buffered channels this may lead to different vector clock annotations.
For unbuffered channels it turns out that the behavior, i.e.~vector clock annotation,
is completely deterministic regardless of the schedule. Formal details follow below.

We write $B^k$ to denote the initial buffer connected to some buffered channel
of size $k>0$. We assume that $B^k$ is filled with $k$ elements $\bot^{\initVC}$.
That is, $B^k=[\bot^{\initVC}, \dots, \bot^{\initVC}]$.

\begin{definition}[Deterministic Replay]
\label{def:det-comm}    
  Let $p$ be a program and $q$ its instrumentation
  where for a specific program run we observe
  the list $[\thread{1}{T_1},\dots,\thread{n}{T_n}]$
  of thread-local traces.
  Let $x_1^{k_1}, \dots, x_m^{k_m}$ be the buffered channels appearing in $p$
  annotated with their buffer size.

  We say that trace replay
    \bda{c}
      \replay{\sep{\thread{x_1^{k_1}}{B^{k_1}},\dots, \thread{x_m^{k_m}}{B^{k_m}}}
               {[\thread{\clock{1}{\initVC}}{T_1}, \dots, \thread{\clock{n}{\initVC}}{T_n}]}\\}
          {E}
          {\\ \sep{\thread{x_1^{k_1}}{B_1},\dots, \thread{x_m^{k_m}}{B_m}}
            {[\thread{\clock{1}{cs_1}}{T_1'},\dots, \thread{\clock{n}{cs_n}}{T_n'}]}}
    \eda
    is \emph{exhaustive} iff no further trace replay rules
    are applicable and for all $i\in \{1,\dots,n\}$, each $T'_i$ only contains pre events.
    
    We say that trace replay is \emph{stuck} iff no further trace replay rules are applicable
    and for some $i\in \{1,\dots,n\}$, $T'_i$ contains some post events.
  
  We say that the list of thread-local traces $[\thread{1}{T_1},\dots,\thread{n}{T_n}]$
  enjoys \emph{deterministic replay}
  if for any two exhaustive trace replays
  \bda{c}
      \replay{\sep{\thread{x_1^{k_1}}{B^{k_1}},\dots, \thread{x_m^{k_m}}{B^{k_m}}}
               {[\thread{\clock{1}{\initVC}}{T_1}, \dots, \thread{\clock{n}{\initVC}}{T_n}]}\\}
          {E}
          {\\ \sep{\thread{x_1^{k_1}}{B_1},\dots, \thread{x_m^{k_m}}{B_m}}
            {[\thread{\clock{1}{cs_1}}{T_1'},\dots, \thread{\clock{n}{cs_n}}{T_n'}]}}
   \eda
 and
 \bda{c}
   \replay{\sep{\thread{x_1^{k_1}}{B^{k_1}},\dots, \thread{x_m^{k_m}}{B^{k_m}}}
               {[\thread{\clock{1}{\initVC}}{T_1}, \dots, \thread{\clock{n}{\initVC}}{T_n}]}\\}
          {E'}
          {\\ \sep{\thread{x_1^{k_1}}{B_1'},\dots, \thread{x_m^{k_m}}{B_m'}}
            {[\thread{\clock{1}{cs_1'}}{T_1''},\dots, \thread{\clock{n}{cs_n'}}{T_n''}]}}
    \eda
   we have that
   vector clocks for events at the same program location in $E$ and $E'$ are identical.
 In case of loops, we compare program locations used at the same instance.
\end{definition}

\begin{proposition}[Deterministic Replay for Unbuffered Channels]
\label{prop:sync-det-replay}  
  Let $p$ be a program consisting of unbuffered channels only.
  Then, any list of thread-local run-time traces obtained
  enjoys deterministic trace replay.
\end{proposition}
\begin{proof}
  Trace replay rules define a rewrite relation among configurations $\sep{Q}{U}$.
  The formulation in Figure~\ref{f:trace-replay} assumes that the list of run-time traces
  can be shuffled so that replay rules only operate on the first, respectively, the first
  and the second element in that list.
  In the following, we assume a more liberal formulation
  of trace replay rules where any trace can be picked to apply a rule.
  Both formulations are equivalent and hence enjoy the same properties.
  But the the more liberal formulation
  allows us to drop rules \rlabel{Shuffle} and \rlabel{Closure} from consideration
  and we apply some standard (rewriting) reasoning method.
  We proceed by showing that the more liberal formulation is terminating
  and locally confluent. 

  Termination is easy to establish as each rule consumes at least one event.
  
  Next, we establish local confluence by observing all critical pairs.
  In our setting, critical pairs include configurations $\sep{Q}{U}$ as well as the
  the vector clock annotated events $E$ obtained during rewriting.
  That is, we observe all situations where for $\sep{Q}{U}$ and a single rewrite step we find
  $\replay{\sep{Q}{U}}{E_1}{\sep{Q_1}{U_1}}$ and
  $\replay{\sep{Q}{U}}{E_2}{\sep{Q_2}{U_2}}$
  for some $E_1, E_2, Q_1, Q_2, U_1, U_2$.
  We show that all critical pairs are joinable by examining all rule combinations
  that lead to a critical pair.
  In our setting, joinable means that we find
    $\replay{\sep{Q_1}{U_1}}{E_1'}{\sep{Q_1'}{U_1'}}$ and
  $\replay{\sep{Q_2}{U_2}}{E_2'}{\sep{Q_2'}{U_2'}}$
  where vector clocks for events at the same program location in
  $E_1 \pp E_1'$ and $E_2 \pp E_2'$ are identical.

  As we only consider unbuffered channels, rules \rlabel{Send} and \rlabel{Receive} are not applicable.
  Hence, we only need to consider combinations of rules \rlabel{Signal/Wait}, \rlabel{Sync},
  \rlabel{Receive-Closed}, \rlabel{Close} and \rlabel{Default}.

  Rules \rlabel{Receive-Closed}, \rlabel{Close} and \rlabel{Default} only affect
  a specific run-time trace and the events in that trace.
  So, any combination of these rules that lead to a critical pair is clearly joinable.

  Rules \rlabel{Signal/Wait} and \rlabel{Sync} affect two run-time traces.
  The choice which two traces are affected is fixed.
  For each (synchronous) sent that took there is exactly one matching receive and vice versa.
  This is guaranteed by our tracing scheme where we identify send-receive pairs
  via the sender's thread id and program counter.
  The same applies to signal and wait.
  Hence, any critical pair that involves any of these two rules is joinable.

  We summarize. The more liberal rules are terminating
  and locally confluent. By Newmann's Lemma we obtain confluence.
  Any derivation with rules in Figure~\ref{f:trace-replay} can be expressed in terms of the more liberal rules.
  Hence, the rules in Figure~\ref{f:trace-replay} are confluent.
  Confluence implies deterministic trace replay.
\end{proof}

Recall Example~\ref{ex:sync-chan} where the run-time trace records that (a) thread~2 synchronizes
with thread~3, and (b) thread~4 synchronizes with thread~5. The actual schedule (if first (a) or (b)) is
not manifested in the trace. Indeed, if (a) or (b) first does not affect the vector clock information obtained.

\begin{proposition}
\label{prop:unbuffered-time-complexity}  
  Let $p$ be a program consisting of unbuffered channels only
  where for some instrumented program run we find that
  $k$ is the number of thread-local run-time traces
  and $m$ is the sum of the length of all thread-local run-time traces.
  Then, vector clock annotated events can be computed in time $O(k^2 * m)$.
\end{proposition}
\begin{proof}
  Proposition~\ref{prop:sync-det-replay}  guarantees that any order in which trace replay
  rules are applied yields the same result and most importantly we never get stuck.
  Rules \rlabel{Signal/Wait} and \rlabel{Sync} need to find two matching partners
  in two distinct traces. All other rules only affect a single trace.
  Hence, each rewriting step requires $O(k^2)$. The number of possible combinations of two elements from
  a set of size $k$.
  Each rewriting step reduces the size of at least one of the thread-local traces.
  Hence, we must obtain the result in $O(m)$ steps.
  So, overall computation of $E$ takes time $O(k^2 * m)$.
\end{proof}


The situation is different for buffered channels.

\begin{example}
\label{ex:buffer-one}
Consider 
\bda{lcl}
    [ x \assign \MAKECHAN{1},
      \\ \GO\ [\SEND{x}{1},            && (1)
        \\ \ \ \ \ \ \ \ \ \ \ \ \ \ \ \RCV{x}],           && (2)
      \\ \SEND{x}{1},          &&  (3)
      \\ \RCV{x}]              &&  (4)
\eda
We assume that the send and receive in the helper thread execute first.
Here is the resulting trace.
\bda{l}
    [\thread{1}{[\signalTrace{2}, \pre{\snd{x}}, \postSnd{1}{3}{x}, \pre{\rcv{x}}, \postRcv{1}{3}{x}]},
\\ \thread{2}{[\waitTrace{2}, \pre{\snd{x}}, \postSnd{2}{1}{x}, \pre{\rcv{x}}, \postRcv{2}{1}{x}]}]                
\eda
For example, we obtain $\postSnd{1}{3}{x}$ as thread 1's program counter is at position 3
after execution of the \MAKE\ and \GO\ statement.

The initial buffer is of the form $[\bot^{[0,0]}]$.
We infer different vector clocks depending which local trace we process first.
If we start with thread~$1$ we derive vector clock $[2,0]$ for location (1),
$[3,0]$ for location (2), $[3,2]$ for location (3) and $[3,3]$ for location (4).
If we start with thread~$2$, we obtain vector clock $[2,0]$ for location (3),
$[3,0]$ for location (4), $[3,2]$ for location (1) and $[3,3]$ for location (2).
\end{example}

We conclude that replay is non-deterministic for buffered channels.
A different schedule possibly implies a different vector clock annotation.

\begin{proposition}
\label{prop:enumerate-tracereplay}  
  Let $p$ be a program consisting of buffered channels
  where for some instrumented program run we find that
  $k$ is the number of thread-local run-time traces
  and $m$ is the sum of the length of all thread-local run-time traces.
  Then, we can enumerate all possible vector clock annotations in time $O(k^{(2 * m)})$.
\end{proposition}
\begin{proof}
  In each step there are $O(k^2)$ choices to consider that might lead to a different result.
  We need a maximum of $O(m)$ steps.
  Hence, exhaustive enumeration takes time $O(k^{(2 * m)})$.
\end{proof}  

In practice, we find rarely cases of an exponential number of schedules.
As we are in the offline setting, we argue that some extra cost is justifiable
to obtain more details about the program's behavior.


We summarize.
The vector clock annotated trace $E$ contains a wealth of information.
In the upcoming section, we will discuss some specific analysis scenarios
where this information can be exploited.
For certain scenarios, the complete trace $E$ is often not necessary
and for vector clocks a more memory-saving representation can be employed as well.

\section{Analysis Scenarios}
\label{sec:analysis}

We consider four scenarios.
\begin{description}
 \item[MP] Message contention.
 \item[SC] Send  on a closed channel.
 \item[DR] Deadlock recovery.
 \item[AC] Alternative communication partners for send/receive pairs.
\end{description}

MP is to identify performance bottlenecks. This method
can also be used to carry out lock contention.
SC spots a bug whereas DR provides hints to a user how to recover from a deadlock.
AC provides general information about the concurrent behavior of message-passing programs.
Below, we describe how we can implement each scenario based on the vector clock
information provided.
Realistic examples for each scenario will be discussed in Section~\ref{sec:experiments}.

\subsection{Message Contention}
\label{sec:message-contention}

\boxfig{f:epoch-send-receive}{Send/Receive Epoch}{
  \bda{c}
  \ruleform{\replay{\sep{C}{E}}{}{\sep{C}{E}}}
   \\
 \\           
  \rlabel{Channel-Init} \
  \replay{\sep{C}{\clock{\initEvt{i}{x}}{cs} : E}}
         {}{\sep{\varState{x}{cs}{cs} : C}{E}}
  \\
  \\
  \rlabel{Receive-Multiple} \
 \myirule{\mbox{$J = \{ i_1, \dots, i_k \} \subseteq \dom{es}$ maximal and non-empty such that $\forall j \in J. cs[j] \not > es[j]$}          
         }
         {
           \replay{\sep{\varState{x}{es}{es'} : C}{\ppclock{\rcvEvt{i}{x}}{cs}{cs'} : E}}
                  {}
                  {\sep{\varState{x}{[\thread{i}{cs[i]},\thread{i_1}{es[i_1]},\dots,\thread{i_k}{es[i_k]}]}{es'} : C}{E}}
         }
  \\
  \\
 \rlabel{Receive-Single} \
 \myirule{cs > es
         }
         {\replay{\sep{\varState{x}{es}{es'} : C}{\ppclock{\rcvEvt{i}{x}}{cs}{cs'} : E}}
                 {}
                 {\sep{\varState{x}{\thread{i}{cs[i]}}{es'} : C}{E}}
         }
  \\
  \\
  \rlabel{Send-Multiple} \
 \myirule{\mbox{$J = \{ i_1, \dots, i_k \} \subseteq \dom{es}$ maximal and non-empty such that $\forall j \in J. cs[j] \not > es[j]$}          
         }
         {
           \replay{\sep{\varState{x}{es'}{es} : C}{\ppclock{\sndEvt{i}{x}}{cs}{cs'} : E}}
                  {}
                  {\sep{\varState{x}{es'}{[\thread{i}{cs[i]},\thread{i_1}{es[i_1]},\dots,\thread{i_k}{es[i_k]}]} : C}{E}}
         }
  \\
  \\  
 \rlabel{Send-Single} \
 \myirule{cs > es
         }
         {\replay{\sep{\varState{x}{es'}{es} : C}{\ppclock{\sndEvt{i}{x}}{cs}{cs'} : E}}
                 {}
                 {\sep{\varState{x}{es'}{\thread{i}{cs[i]}} : C}{E}}
         }
  \eda
}  

We wish to check if there are competing send operations for a specific channel $x$.
If we also take into account dangling events, we simply consult $E$
and count all events $\ppclock{\sndEvt{n_i}{x}}{cs_i}{cs_i'}$
where for $i \not= j$ we have that
%
pre vector clocks $cs_i$ and $cs_j$ are incomparable (the events are concurrent).
The same check applies to receive events.

To carry out the analysis efficiently it is unnecessary to
construct the entire set $E$.
For each channel $x$, we only need to keep track of concurrent sends/receives.
For each concurrent operation, instead of the full (pre) vector clock, we only record
a pair of thread id and the time stamp for that thread.
We refer to this pair as an epoch
following~\cite{flanagan2010fasttrack}.

\begin{definition}[Send/Receive Epoch]
  \bda{lcll}
    e & ::= & \thread{i}{n} & \mbox{Epoch}
  \\
  es & ::= & [] \mid e : es
  \\
  C & ::= & [] \mid \varState{x}{es}{es} : V   & \mbox{Senders/Receivers}
\eda
\end{definition}
%

Notation $es$ denotes a list of epochs. 
In the extreme case, $es$ denotes the entire vector clock.
For example, the channel's initial state corresponds to the vector clock
at the declaration site.
If $es$ covers all threads, we treat $es$ as the 
sorted (according to thread ids) list $[\thread{1}{n_1},\dots,\thread{k}{n_k}]$
and consider $es$ as equivalent to $[n_1,\dots,n_k]$.

Let $es = [\thread{i_1}{n_1},\dots,\thread{i_k}{n_k}]$.
We define $\dom{es} = \{ i_1,\dots,i_k \}$
and $es[i_l] = n_l$ for $i_l \in \dom{es}$.
We define $cs > es$ if for all $i \in \dom{es}$ we have that $cs[i] > es[i]$.
We define $cs \not> es$ if there exists $i \in \dom{es}$ such that $cs[i] \not> es[i]$.

For convenience, we carry out the epoch optimization on the annotated trace $E$ instead of
adjusting the rules in Figure~\ref{f:trace-replay}.
In practice, the epoch optimization can be integrated into the trace replay rules.
The epoch optimization rules in Figure~\ref{f:epoch-send-receive} introduce a rewrite relation among $\sep{C}{E}$
where for brevity we omit auxiliary rules to drop irrelevant events (close and default) and rearrange $C$ suitably.
Rule \rlabel{Channel-Init} deals with the initialization of a channel
and assumes that the declaration site is instrumented
such that we can obtain $x$'s initial vector clock. For this purpose,
we assume an event $\clock{\initEvt{i}{x}}{cs}$.

Rule \rlabel{Receive-Multiple} covers multiple concurrent receives.
We build the maximum number of receives that are concurrent
to the current event.
Rule \rlabel{Receive-Single} covers the case of a receive that happens after all prior receives.
The formulation for send events is orthogonal to receives.
See rules \rlabel{Send-Multiple} and \rlabel{Send-Single}.

High message contention means that for $\varState{x}{es_1}{es_2}$
there is high number of elements in either $es_1$ or $es_2$.

\subsection{Send on Closed}
\label{sec:send-closed}

We wish to check if there is a schedule where a send operation
attempts to transmit to a closed channel. We assume that this bug did not arise
for the given program run. Hence, we use the vector clock information to test for
a send operation that either succeeds or is concurrent to a close operation.
In terms of $E$, we check for
 events $\ppclock{\sndEvt{i}{x}}{cs_1}{cs_2}$, $\clock{\closeEvt{j}{x}}{cs_3}$
where either $cs_1$ succeeds $cs_3$ or $cs_1$ and $cs_3$ are incomparable.

The epoch optimization applies here as well.
We maintain the list of concurrent send operations as in case of message contention.
Instead of the list of receives, we only record the vector clock of the close operation.
Each time we update the list of sends, we check that none of the send epochs succeeds
or is concurrent to the close.

As discussed in Section~\ref{sec:properties}, the vector clock annotation
obtained for events may be non-deterministic.
Hence, we may miss a send on closed channel bug depending on the schedule.
The advantage of our method is to explore alternative schedules to reveal
such hidden bugs.

\begin{example}
  \label{ex:send-on-closed-two}
  Consider
\bda{lcl}
    [ x \assign \MAKECHAN{1}, 
      \\      \GO\ [\SEND{x}{1}, \RCV{x},\CLOSE{x}],
      \\ \SEND{x}{1}, \RCV{x}]

\eda
We assume a program run where first the main thread executes and then the other thread.
This yields the following run-time trace.
\bda{l}
    [\thread{1}{[\signalTrace{2},  \pre{\snd{x}}, \postSnd{1}{3}{x}, \pre{\rcv{x}}, \postRcv{1}{3}{x}]},
      \\  \thread{2}{[\waitTrace{2}, \pre{\snd{x}}, \postSnd{2}{1}{x}, \pre{\rcv{x}}, \postRcv{2}{1}{x},\postClose{x}]}]
\eda
In our approach, we can explore a different schedule by processing thread~2 (after processing of signal/wait).
Then, the close operation appears to be concurrent to the send in the main thread.
\end{example}

\subsection{Alternative Communications}

We refer to a \emph{match pair} as a pair of vector clock
annotated events $(e_1,e_2)$ where $e_1$ is a sender and $e_2$
is the matching receiver over a common channel $x$.
For an unbuffered (synchronous) channel,
we have that $e_1 = \ppclock{\sndEvt{i}{x}}{cs_1}{cs}$
and $e_2 = \ppclock{\rcvEvt{j}{x}}{cs_2}{cs}$.
That is, their post vector clocks are synchronized.
For a buffered channel,
we have that $e_1 = \ppclock{\sndEvt{i}{x}}{cs_1}{cs_1'}$
and $e_2 = \ppclock{\rcvEvt{j}{i}{x}}{cs_2}{cs_2'}$
where $cs_2' = \maxC{\incC{j}{cs}}{cs_1'}$.

Match pairs can be directly computed during trace replay
as their underlying post events are uniquely connected via
the sender's thread id and program counter.
We assume $k$ is the number of thread-local run-time traces
and $m$ is the sum of the length of all thread-local run-time traces.
For each match pair we compute \emph{alternative communications}.
For each sender, we count the number of concurrent receives
and for each receiver the number of concurrent sends where
for each candidate we test if the pre vector clocks are incomparable.
We assume the comparison test among vector clocks takes constant time.
That is, $O(k) = O(1)$.
For each match pair, there can be at most $O(m)$ candidates.
Hence, computations of alternative communications for a specific trace replay
run takes time $O(m * m)$.
Our experiments show that alternatives can be computed efficiently as 
we can use the thread's vector clock to prune the search space
for candidates.

Via a similar method, we can compute the number of alternatives for
each \emph{not selected} case of a \SELECT\ statement.
We refer to this analysis scenario as {\bf ASC}.

\subsection{Deadlock Recovery}
\label{sec:deadlock-recovery}

We consider the scenario where program execution results in a deadlock.
Via a similar method as described for alternative communications,
we can search for potential partners for dangling events.

\begin{example}
  Recall the example from the introduction.
  \bda{lcl}
      [ x \assign \SYNCMAKECHAN,
        \\ \GO\ [\SEND{x}{1}],  && (1)
        \\ \GO\ [\RCV{x}],  && (2)
        \\ \SEND{x}{1} ]  && (3)        
  \eda
  We assume a deadlock because (1) synchronizes with (2).
  Based on the pre vector clock information of the event resulting from (3),
  we can feedback to the user that the deadlock could possibly be resolved
  assuming (3) synchronizes with (1).
\end{example}  

There are cases where no alternatives can be provided.

\begin{example}
  Consider the classic example of a deadlock due to reversed `lock' order
  where we model a mutex via a buffered channel.  
  \bda{lcl}
      [ x \assign \MAKECHAN{1}, y \assign \MAKECHAN{1},
        \\ \GO\ [\SEND{y}{1}, \SEND{x}{1}, \RCV{x}, \RCV{y}],
        \\ \SEND{x}{1}, \SEND{y}{1}, \RCV{y}, \RCV{x}]
  \eda
  Our analysis reports that no alternatives exist.
\end{example}

The interpretation of analysis is left to the user.
We believe the information provided are highly useful
in gaining further insights into the (deadlock) bug.

\section{Experiments}
\label{sec:experiments}


We have a built a  prototype in Go.
A snapshot of our implementation including all examples used for experimentation can be accessed via
\begin{verbatim}
    https://github.com/KaiSta/gopherlyzer-GoScout
\end{verbatim}
Our implementation includes analysis methods and optimizations discussed in the earlier sections.
Experiments are conducted on a Intel i7 6600U with 12GB RAM, SSD and Windows 10.
The results are shown in Figure \ref{f:testResults}.

\subsection{Implementation}

The toolchain, entirely implemented in Go, consists of three parts.
(1) Instrumentation. (2) Execution. (3) Analysis.
In the first part, we instrument the source code  to emit pre and post events.
As we need to provide an additional argument for channels and channel operations,
we need to access the channel's type.
For this purpose,  we make use of go-parser to  update the AST by updating the channel's type to an anonymous struct that
contains a field for the necessary thread information and the original type of the channel as a value field.
The instrumentation follows the scheme outline in Figure~\ref{f:instrumentation}.
Additionally the main function is instrumented to start and stop the tracer.
This is necessary to ensure that all occurred events are written to the trace since Go programs terminate as soon as the main thread exits.
The tracer uses a separate thread that receives the events through a buffered channel and writes them to the trace which is necessary in case of a deadlock in the program.
The separate thread will write all events stored in the channel buffer to the trace before it waits for new messages which will trigger the occurred deadlock with a small delay.
After execution, we apply the trace replay method described previously.

\subsection{Examples}

For experimentation, we use the following examples.
They consist of some real-world examples as well as our own examples
to highlight certain aspects of our approach.

\subsubsection*{pgzip} is a parallel gzip compression/decompression written in Go. It splits the file in several blocks that are send through a buffered channel where the worker threads can collect and compress them. After compressing a block it is send through another channel to a thread that collects the blocks to write them to a file. pgzip makes typical use of synchronous and asynchronous channels to either transfer data or to send signals to other threads like `abort' by closing specific channels. We intend to test if the collection always has to happen in a fixed order and for sends on closed channels due to the way `aborts' are implemented. For the test we compress a 8mb file. (\url{https://github.com/klauspost/pgzip})

\subsubsection*{htcat} performs parallel, pipelined executions of a single HTTP `GET' to improve the download speed. It distributes the work to multiple threads that perform a part of the download and collects the finished blocks in a separate thread. It uses a different scheme to recollect the blocks that we test with our analyses. For the test a 8mb download was used. (\url{https://github.com/htcat/htcat})

\subsubsection*{go-dsp/fft} is a digital signal processing package for Go where we test the parallel FFT implementation. Due to the intensive use of channels it produces huge traces with a medium amount of alternative communications. For the test we used the unit tests delivered with the package.(\url{https://github.com/mjibson/go-dsp})

\subsubsection*{go-hashmap} is included in the Go programming language since version 1.9. It is a thread safe alternative to the previous map implementation. It is optimized for two uses cases where the first is that keys are written once but read many times by different threads and second for multiple threads that each work with their own key without touching the keys of another thread. For all other use cases it becomes too slow because of the lock contention that occurs. We use our analyses to make this lock contention visible. For this we use 3 writer and two reader threads that work on the same 100 keys.(\url{https://golang.org/src/sync/map.go})

\subsubsection*{newsreader} is a artificial example for a program that runs into a deadlock with detectable alternative communications that indicate that the deadlock might not be the only possible outcome of running this program.

\subsubsection*{cyclic} contains a cyclic dependency for the messages that are send between the threads. Hence it runs into a deadlock which is unavoidable and therefore our analysis will not report any alternative communications.

\subsection{Tracing Overhead}

The tracing overhead is between 2 and 41\%.
The two real world programs pgzip and htcat have the lowest overhead since they spend most of their time with reading and writing files. go-dsp/fft and the go hashmap have both an overhead of around 40\% because of their intensive use of channels and/or locks.

\subsection{Message and Lock Contention}

Our analyses shows for the concurrent hashmap over 13000 message contention situations that can only occur for the included locks that allow a single thread to make a safe read or write. This shows that nearly every access needs to use the slow path using locks to complete its tasks which results in a huge slowdown. 

\subsection{Alternative Communications}

The pgzip and htcat examples use worker threads that collect the completed blocks and put them together. They must make sure that the blocks are collected in a specific order to be able to put them back together. It would have been an error to see alternative communications for those specific receives that collect the blocks. In both cases the important receives can always only receive from a single thread at a time and therefore will always put the data back together in the right order.

\subsection{Send on Closed}

All tested real world examples that use the close operations are not prone to this kind of error in the current state. This is mostly because of their use of a separate channel that is never used except for closing it to signal another thread. For our set of real world examples only pgzip closes a channel that was previously used to send and receive data.

\subsection{Deadlock Recovery}

For the cyclic dependency we used the standard double lock example where two threads lock two different locks in different orders. As our analysis shows their are no alternative matching partners but 2 cases of message contention since both threads are concurrent alternatives for each locking operation.
The newsreader example on the other hand has 10 alternative matching partners and 4 cases of message contention during the run. These show that the first call to collect the news consumes both messages since it has all the possible alternatives while the second call is stuck.

\subsection{Results}

The behavior of our test cases is largely independent of the schedule chosen.
Therefore, results are reported for a specific schedule (vector clock annotation).
For each example, we accumulate the number of alternatives found. The information that four alternative communications (AC) were found means that we either detected four alternatives for a single communication pair or for two communication pairs each two alternative communications for example. The same is true for MP and ASC. SC counts the sends that might occur on a closed channels and DR if it is a analysis on a deadlocked program. For the additional field ASC (Alternative Select Case) each communication pair where one communication partner is a select, we count the amount of alternative select cases available.

\boxfig{f:testResults}{Test Results}{	
	\begin{minipage}{\textwidth}
	
	\centering
	\begin{tabular}{l|l|l|l|l|l|l|l|l}
		\textbf{Program}    & \textbf{LOC}  & \textbf{Trace Size} & \textbf{Time}     & \textbf{AC}      & \textbf{MP}    & \textbf{ASC} & \textbf{SC} & \textbf{DR}   \\
		pgzip      & 1201 & 468        & 37ms     & 1386    & 0     & 55  & 0  &      \\
		htcat      & 728  & 4263       & 186ms    & 2034    & 86    & 21  & 0  &      \\
		go-dsp/fft & 843  & 99087      & 516.8sec & 8458    & 2558  & 0   & 0  &      \\
		go-hashmap & 421  & 19363      & 9713ms   & 1443130 & 13039 & 0   & 0  &      \\
		newsreader &  27  &    28        &    3ms      &    10     &   4    &    0 &  0  & true \\
		cyclic     &  25  &     10      &   4ms   &   0      &   2    &  0   &  0  & true
	\end{tabular}
\end{minipage}
}

\section{Related Work}
\label{sec:related-work}

\mbox{} \\
\noindent
{\bf Verification of Go programs.}
There are several recent works that consider
the static verification of Go programs.
Work by Ng and Yoshida~\cite{DBLP:conf/cc/NgY16}
and our own prior work in collaboration with Thiemann~\cite{DBLP:conf/aplas/StadtmullerST16}
considers deadlock detection.
In addition to safety properties such as deadlocks,
the work by Lange, Ng, Toninho and Yoshida~\cite{DBLP:conf/popl/LangeNTY17}
also considers liveness properties.
Recent work by Lange, Ng, Toninho and Yoshida~\cite{Lange:2018:SVF:3180155.3180157}
employs behavioral types to capture an even richer set
of liveness and safety properties.
A common issue with static analysis is scalability.
Experiments reported in~\cite{Lange:2018:SVF:3180155.3180157} only cover programs
with a small portion of the program related to concurrency.

Dynamic analysis are more likely to scale to real-world programs.
In the Go context, we are only aware of two works that support the dynamic analysis of concurrent Go.
The Go programming language includes a dynamic data race detection tool~\cite{gorace}
but does not offer any form of analysis for message-passing like we discuss here.
Our own work~\cite{DBLP:conf/hvc/SulzmannS17} introduces a trace-based method to observe
the run-time behavior of (synchronous) message-passing Go programs.
In this work, we include the proper treatment of buffered channels and close operations.
We introduce the idea of pre vector clocks and consider
several analysis scenarios including the epoch optimization.
A further significant difference to the present work is that in~\cite{DBLP:conf/hvc/SulzmannS17} we
derive a dependency graph to capture the happens-before relation among events.
The dependency graph appears to be less efficient and precise compared to vector clocks.

\mbox{} \\
\noindent
{\bf Actor model~\cite{Agha:1986:AMC:7929}.}
The actor model supports a more restricted form of message-passing where channels
are associated to actors (mailboxes).
Hence, there can be multiple senders but only a single receiver.
The work summarized in~\cite{DBLP:conf/rv/CassarFAAI17} discusses a series of tools for Erlang~\cite{Armstrong:2013:PES:2566708}
to monitor properties specified in a dynamic logic.
Further works in the actor setting consider systematic testing methods such that (unit) tests catch a bug.
For example, see \cite{Tasharofi:2013:BCA:3107656.3107674,DBLP:conf/cc/AlbertGI16,DBLP:conf/icst/ChristakisGS13}.
Our focus is to examine in detail a specific execution run
for which we provide several analysis scenarios.

\mbox{} \\
\noindent
{\bf Message-passing interface (MPI)~\cite{Forum:1994:MMI:898758}.}
MPI supports an asynchronous form of receive.
This leads to issues when using vector clocks (happens-before relation).
See \cite{DBLP:conf/IEEEpact/VoGKSSB11,Vo:2011:SFD:2231450} for an in-depth discussion.
In Go, receive must either synchronize via a sender or via a buffer.
Hence, the use of vector clocks poses no problem.



\mbox{} \\
\noindent
{\bf Message-passing a la CSP~\cite{Hoare:1978:CSP:359576.359585}.}
 The work by Fidge~\cite{fidge1988timestamps} and Mattern~\cite{Mattern89virtualtime}
 shows how to compute vector clocks in the message-passing setting.
 Besides vector clocks for post (committed) events, we introduce
 the idea of vector clocks for pre events (possibly can commit)
 and give a precise treatment of buffered channels.
 Both extensions can lead to improved analysis results.

Netzer~\cite{DBLP:conf/lcpc/Netzer93,DBLP:conf/sc/NetzerM92} traces
the event order at run-time via Fidge-style vector clocks
in the context of message-passing.
However, it is unclear if his system is able to support multiple channels of different
kinds (synchronous/asynchronous) as supported in our approach.


Ronsse and co-workers~\cite{Ronsse:2003:DSM:860016.860024,Ronsse:1999:RFI:312203.312214}
uses a two-level tracing approach.
First, the program is instrumented to trace synchronization points based on thread-local timestamps.
Based on this trace information, the program is then instrumented such that the previous trace can be
recreated during execution. During this (replay) execution run, vector clocks are used to
infer the order among events.
The approach we propose is more flexible in that based on a single trace, alternative schedules (vector clock annotations)
can be derived. Thus, we can identify problems that might not be obvious
based on the actual program run.

\mbox{} \\
\noindent
{\bf Run-time tracing.}
We make use of a fully automatic light-weight instrumentation and tracing scheme
derived from our own prior work~\cite{DBLP:conf/hvc/SulzmannS17}.
Earlier work uses frameworks~\cite{Kiczales:2001:OA:646158.680006}
that instrument at the byte-code level
or require to adapt the run-time.
Unlike the work of Bhansali and others \cite{Bhansali:2006:FIT:1134760.1220164},
we do not include time stamp information during tracing.




\mbox{} \\
\noindent
{\bf Weakening happens-before and predictive analysis.}
A well-known issue with vector clocks is that the happens-before
relation obtained is tied to a specific execution run (trace).
Hence, we might miss a bug in our program that would become obvious
assuming an alternative schedule, i.e.~suitable reordering of the trace.

In the context of shared memory, there are works
that weaken the happens-before relation~\cite{Smaragdakis:2012:SPR:2103621.2103702,Kini:2017:DRP:3140587.3062374},
pursue alternative schedules in parallel~\cite{DBLP:journals/sttt/SenRA06},
consider all possible reorderings~\cite{Huang:2014:MSP:2666356.2594315}
and take into account source code information~\cite{Wang:2009:SPA:1693345.1693367}.
An issue in the shared memory setting is to maintain the write-read dependency
among shared variables to avoid false positives.
If due to rescheduling a different value could be read, the analysis is non-predictive
as such a program run may never be possible.
For example, the work by Huang, Luo and Rosu~\cite{Huang:2015:GGP:2818754.2818856} traces
values to guarantee that write-read dependencies are respected.

We employ the sender's thread id and program counter to establish a unique connection between sender and receiver.
This guarantees predictability of our analysis.
The tracing overhead is fairly low as supported by our experimental results.
Any vector clock annotation obtained
via trace replay corresponds to a valid schedule.
Hence, we can avoid the bias towards a specific execution run (trace)
by employing thread-local traces.
We can explore alternative schedules by enumerating all possible schedules/vector clock
annotations during trace replay.

In the MPI setting, the works~\cite{Forejt:2014:PPA:2962288.2962307,YuHuang16}
use model checking ideas to predict devious schedules based on a single trace.
As mentioned above, some variant of MPI's receive operation differs from Go.
Furthermore, MPI demands that each send is directed to a specific thread only.
This certainly makes the MPI problem more feasible as the search space for possible schedules
is reduced.



\section{Conclusion}
\label{sec:conclusion}

We employ a two-phase method for the dynamic analysis of message-passing Go programs.
The first phase, instrumentation and tracing, has a fairly low run-time overhead which is supported
by our experiments.
The second phase analyzes the recorded traces and  recovers vector clock information.
The analysis phase can be tailored to find bugs and identify performance bottlenecks.
A key feature of our approach is the use of thread-local traces.
Thus, we can observe behavior that might result from alternative schedule.
In future work, we plan to investigate further heuristics to detect devious schedules
and additional user scenarios to exploit the information inferred.

\begin{acks}
We thank some ISSTA'18 and PPDP'18 reviewers for their comments.

\end{acks}


\bibliography{main}


\begin{thebibliography}{37}


\ifx \showCODEN    \undefined \def \showCODEN     #1{\unskip}     \fi
\ifx \showDOI      \undefined \def \showDOI       #1{#1}\fi
\ifx \showISBNx    \undefined \def \showISBNx     #1{\unskip}     \fi
\ifx \showISBNxiii \undefined \def \showISBNxiii  #1{\unskip}     \fi
\ifx \showISSN     \undefined \def \showISSN      #1{\unskip}     \fi
\ifx \showLCCN     \undefined \def \showLCCN      #1{\unskip}     \fi
\ifx \shownote     \undefined \def \shownote      #1{#1}          \fi
\ifx \showarticletitle \undefined \def \showarticletitle #1{#1}   \fi
\ifx \showURL      \undefined \def \showURL       {\relax}        \fi
\providecommand\bibfield[2]{#2}
\providecommand\bibinfo[2]{#2}
\providecommand\natexlab[1]{#1}
\providecommand\showeprint[2][]{arXiv:#2}

\bibitem[\protect\citeauthoryear{Agha}{Agha}{1986}]%
        {Agha:1986:AMC:7929}
\bibfield{author}{\bibinfo{person}{Gul Agha}.} \bibinfo{year}{1986}\natexlab{}.
\newblock \bibinfo{booktitle}{{\em Actors: A Model of Concurrent Computation in
  Distributed Systems}}.
\newblock \bibinfo{publisher}{MIT Press}, \bibinfo{address}{Cambridge, MA,
  USA}.
\newblock
\showISBNx{0-262-01092-5}


\bibitem[\protect\citeauthoryear{Albert, G{\'{o}}mez{-}Zamalloa, and
  Isabel}{Albert et~al\mbox{.}}{2016}]%
        {DBLP:conf/cc/AlbertGI16}
\bibfield{author}{\bibinfo{person}{Elvira Albert}, \bibinfo{person}{Miguel
  G{\'{o}}mez{-}Zamalloa}, {and} \bibinfo{person}{Miguel Isabel}.}
  \bibinfo{year}{2016}\natexlab{}.
\newblock \showarticletitle{{SYCO:} a systematic testing tool for concurrent
  objects}. In \bibinfo{booktitle}{{\em Proc.\ of CC'16}}.
  \bibinfo{publisher}{{ACM}}, \bibinfo{pages}{269--270}.
\newblock


\bibitem[\protect\citeauthoryear{Armstrong}{Armstrong}{2013}]%
        {Armstrong:2013:PES:2566708}
\bibfield{author}{\bibinfo{person}{Joe Armstrong}.}
  \bibinfo{year}{2013}\natexlab{}.
\newblock \bibinfo{booktitle}{{\em Programming Erlang: Software for a
  Concurrent World}}.
\newblock \bibinfo{publisher}{Pragmatic Bookshelf}.
\newblock
\showISBNx{193778553X, 9781937785536}


\bibitem[\protect\citeauthoryear{Bhansali, Chen, de~Jong, Edwards, Murray,
  Drini\'{c}, Miho\v{c}ka, and Chau}{Bhansali et~al\mbox{.}}{2006}]%
        {Bhansali:2006:FIT:1134760.1220164}
\bibfield{author}{\bibinfo{person}{Sanjay Bhansali}, \bibinfo{person}{Wen-Ke
  Chen}, \bibinfo{person}{Stuart de Jong}, \bibinfo{person}{Andrew Edwards},
  \bibinfo{person}{Ron Murray}, \bibinfo{person}{Milenko Drini\'{c}},
  \bibinfo{person}{Darek Miho\v{c}ka}, {and} \bibinfo{person}{Joe Chau}.}
  \bibinfo{year}{2006}\natexlab{}.
\newblock \showarticletitle{Framework for Instruction-level Tracing and
  Analysis of Program Executions}. In \bibinfo{booktitle}{{\em Proc.\ of VEE
  '06}}. \bibinfo{publisher}{ACM}, \bibinfo{address}{New York, NY, USA},
  \bibinfo{pages}{154--163}.
\newblock
\showISBNx{1-59593-332-8}
\showDOI{%
\url{https://doi.org/10.1145/1134760.1220164}}


\bibitem[\protect\citeauthoryear{Cassar, Francalanza, Attard, Aceto, and
  Ing{\'{o}}lfsd{\'{o}}ttir}{Cassar et~al\mbox{.}}{2017}]%
        {DBLP:conf/rv/CassarFAAI17}
\bibfield{author}{\bibinfo{person}{Ian Cassar}, \bibinfo{person}{Adrian
  Francalanza}, \bibinfo{person}{Duncan~Paul Attard}, \bibinfo{person}{Luca
  Aceto}, {and} \bibinfo{person}{Anna Ing{\'{o}}lfsd{\'{o}}ttir}.}
  \bibinfo{year}{2017}\natexlab{}.
\newblock \showarticletitle{A Suite of Monitoring Tools for Erlang}. In
  \bibinfo{booktitle}{{\em Proc.\ of RV-CuBES'17}} {\em (\bibinfo{series}{Kalpa
  Publications in Computing})}, Vol.~\bibinfo{volume}{3}.
  \bibinfo{publisher}{EasyChair}, \bibinfo{pages}{41--47}.
\newblock


\bibitem[\protect\citeauthoryear{Christakis, Gotovos, and Sagonas}{Christakis
  et~al\mbox{.}}{2013}]%
        {DBLP:conf/icst/ChristakisGS13}
\bibfield{author}{\bibinfo{person}{Maria Christakis}, \bibinfo{person}{Alkis
  Gotovos}, {and} \bibinfo{person}{Konstantinos Sagonas}.}
  \bibinfo{year}{2013}\natexlab{}.
\newblock \showarticletitle{Systematic Testing for Detecting Concurrency Errors
  in Erlang Programs}. In \bibinfo{booktitle}{{\em Proc.\ of ICST'13}}.
  \bibinfo{publisher}{{IEEE} Computer Society}, \bibinfo{pages}{154--163}.
\newblock


\bibitem[\protect\citeauthoryear{Fidge}{Fidge}{1988}]%
        {fidge1988timestamps}
\bibfield{author}{\bibinfo{person}{Colin~J. Fidge}.}
  \bibinfo{year}{1988}\natexlab{}.
\newblock \showarticletitle{Timestamps in message-passing systems that preserve
  the partial ordering}.
\newblock \bibinfo{journal}{{\em Proceedings of the 11th Australian Computer
  Science Conference\/}} \bibinfo{volume}{10}, \bibinfo{number}{1}
  (\bibinfo{year}{1988}), \bibinfo{pages}{56--66}.
\newblock
\showURL{%
\url{http://sky.scitech.qut.edu.au/~fidgec/Publications/fidge88a.pdf}}


\bibitem[\protect\citeauthoryear{Fidge}{Fidge}{1992}]%
        {Fidge:1991:PAT:646210.683620}
\bibfield{author}{\bibinfo{person}{Colin~J. Fidge}.}
  \bibinfo{year}{1992}\natexlab{}.
\newblock \showarticletitle{Process Algebra Traces Augmented with Causal
  Relationships}. In \bibinfo{booktitle}{{\em Proc.\ of FORTE '91}}.
  \bibinfo{publisher}{North-Holland Publishing Co.},
  \bibinfo{address}{Amsterdam, The Netherlands, The Netherlands},
  \bibinfo{pages}{527--541}.
\newblock
\showISBNx{0-444-89402-0}
\showURL{%
\url{http://dl.acm.org/citation.cfm?id=646210.683620}}


\bibitem[\protect\citeauthoryear{Flanagan and Freund}{Flanagan and
  Freund}{2010}]%
        {flanagan2010fasttrack}
\bibfield{author}{\bibinfo{person}{Cormac Flanagan} {and}
  \bibinfo{person}{Stephen~N Freund}.} \bibinfo{year}{2010}\natexlab{}.
\newblock \showarticletitle{FastTrack: efficient and precise dynamic race
  detection}.
\newblock \bibinfo{journal}{{\it Commun. ACM}} \bibinfo{volume}{53},
  \bibinfo{number}{11} (\bibinfo{year}{2010}), \bibinfo{pages}{93--101}.
\newblock


\bibitem[\protect\citeauthoryear{Forejt, Kroening, Narayanaswamy, and
  Sharma}{Forejt et~al\mbox{.}}{2014}]%
        {Forejt:2014:PPA:2962288.2962307}
\bibfield{author}{\bibinfo{person}{Vojt\u{e}ch Forejt}, \bibinfo{person}{Daniel
  Kroening}, \bibinfo{person}{Ganesh Narayanaswamy}, {and}
  \bibinfo{person}{Subodh Sharma}.} \bibinfo{year}{2014}\natexlab{}.
\newblock \showarticletitle{Precise Predictive Analysis for Discovering
  Communication Deadlocks in MPI Programs}. In \bibinfo{booktitle}{{\em Proc.\
  of FM'14}}. \bibinfo{publisher}{Springer}, \bibinfo{pages}{263--278}.
\newblock
\showISBNx{978-3-319-06409-3}


\bibitem[\protect\citeauthoryear{Forum}{Forum}{1994}]%
        {Forum:1994:MMI:898758}
\bibfield{author}{\bibinfo{person}{Message~P Forum}.}
  \bibinfo{year}{1994}\natexlab{}.
\newblock \bibinfo{booktitle}{{\em MPI: A Message-Passing Interface Standard}}.
\newblock \bibinfo{type}{{T}echnical {R}eport}. \bibinfo{address}{Knoxville,
  TN, USA}.
\newblock


\bibitem[\protect\citeauthoryear{Go}{Go}{2018}]%
        {golang}
Go \bibinfo{year}{2018}\natexlab{}.
\newblock \bibinfo{title}{The {Go} Programming Language}.
\newblock \bibinfo{howpublished}{https://golang.org/}.
  (\bibinfo{year}{2018}).
\newblock
\newblock
\shownote{Accessed: 2018-05-02.}


\bibitem[\protect\citeauthoryear{GoRace}{GoRace}{2018}]%
        {gorace}
GoRace \bibinfo{year}{2018}\natexlab{}.
\newblock \bibinfo{title}{Data Race Detector}.
\newblock
  \bibinfo{howpublished}{https://golang.org/doc/articles/race\_detector.html}.
   (\bibinfo{year}{2018}).
\newblock
\newblock
\shownote{Accessed: 2018-05-02.}


\bibitem[\protect\citeauthoryear{Hoare}{Hoare}{1978}]%
        {Hoare:1978:CSP:359576.359585}
\bibfield{author}{\bibinfo{person}{Charles A.~R. Hoare}.}
  \bibinfo{year}{1978}\natexlab{}.
\newblock \showarticletitle{Communicating Sequential Processes}.
\newblock \bibinfo{journal}{{\em Commun. ACM\/}} \bibinfo{volume}{21},
  \bibinfo{number}{8} (\bibinfo{date}{Aug.} \bibinfo{year}{1978}),
  \bibinfo{pages}{666--677}.
\newblock
\showISSN{0001-0782}
\showURL{%
\url{http://doi.acm.org/10.1145/359576.359585}}


\bibitem[\protect\citeauthoryear{Huang, Luo, and Rosu}{Huang
  et~al\mbox{.}}{2015}]%
        {Huang:2015:GGP:2818754.2818856}
\bibfield{author}{\bibinfo{person}{Jeff Huang}, \bibinfo{person}{Qingzhou Luo},
  {and} \bibinfo{person}{Grigore Rosu}.} \bibinfo{year}{2015}\natexlab{}.
\newblock \showarticletitle{GPredict: Generic Predictive Concurrency Analysis}.
  In \bibinfo{booktitle}{{\em Proc.\ of ICSE '15}}. \bibinfo{publisher}{IEEE
  Press}, \bibinfo{pages}{847--857}.
\newblock


\bibitem[\protect\citeauthoryear{Huang, Meredith, and Rosu}{Huang
  et~al\mbox{.}}{2014}]%
        {Huang:2014:MSP:2666356.2594315}
\bibfield{author}{\bibinfo{person}{Jeff Huang}, \bibinfo{person}{Patrick~O'Neil
  Meredith}, {and} \bibinfo{person}{Grigore Rosu}.}
  \bibinfo{year}{2014}\natexlab{}.
\newblock \showarticletitle{Maximal Sound Predictive Race Detection with
  Control Flow Abstraction}.
\newblock \bibinfo{journal}{{\em SIGPLAN Not.\/}} \bibinfo{volume}{49},
  \bibinfo{number}{6} (\bibinfo{date}{June} \bibinfo{year}{2014}),
  \bibinfo{pages}{337--348}.
\newblock
\showISSN{0362-1340}
\showURL{%
\url{http://doi.acm.org/10.1145/2666356.2594315}}


\bibitem[\protect\citeauthoryear{Huang}{Huang}{2016}]%
        {YuHuang16}
\bibfield{author}{\bibinfo{person}{Yu Huang}.} \bibinfo{year}{2016}\natexlab{}.
\newblock {\em \bibinfo{title}{An Analyzer for Message Passing Programs}}.
\newblock \bibinfo{thesistype}{Ph.D. Dissertation}. \bibinfo{school}{Brigham
  Young University}.
\newblock


\bibitem[\protect\citeauthoryear{Kiczales, Hilsdale, Hugunin, Kersten, Palm,
  and Griswold}{Kiczales et~al\mbox{.}}{2001}]%
        {Kiczales:2001:OA:646158.680006}
\bibfield{author}{\bibinfo{person}{Gregor Kiczales}, \bibinfo{person}{Erik
  Hilsdale}, \bibinfo{person}{Jim Hugunin}, \bibinfo{person}{Mik Kersten},
  \bibinfo{person}{Jeffrey Palm}, {and} \bibinfo{person}{William~G. Griswold}.}
  \bibinfo{year}{2001}\natexlab{}.
\newblock \showarticletitle{An Overview of AspectJ}. In
  \bibinfo{booktitle}{{\em Proc.\ of ECOOP'01}}.
  \bibinfo{publisher}{Springer-Verlag}, \bibinfo{address}{London, UK, UK},
  \bibinfo{pages}{327--353}.
\newblock
\showISBNx{3-540-42206-4}


\bibitem[\protect\citeauthoryear{Kini, Mathur, and Viswanathan}{Kini
  et~al\mbox{.}}{2017}]%
        {Kini:2017:DRP:3140587.3062374}
\bibfield{author}{\bibinfo{person}{Dileep Kini}, \bibinfo{person}{Umang
  Mathur}, {and} \bibinfo{person}{Mahesh Viswanathan}.}
  \bibinfo{year}{2017}\natexlab{}.
\newblock \showarticletitle{Dynamic Race Prediction in Linear Time}.
\newblock \bibinfo{journal}{{\em SIGPLAN Not.\/}} \bibinfo{volume}{52},
  \bibinfo{number}{6} (\bibinfo{date}{June} \bibinfo{year}{2017}),
  \bibinfo{pages}{157--170}.
\newblock
\showISSN{0362-1340}
\showDOI{%
\url{https://doi.org/10.1145/3140587.3062374}}


\bibitem[\protect\citeauthoryear{Lamport}{Lamport}{1978}]%
        {lamport1978time}
\bibfield{author}{\bibinfo{person}{Leslie Lamport}.}
  \bibinfo{year}{1978}\natexlab{}.
\newblock \showarticletitle{Time, clocks, and the ordering of events in a
  distributed system}.
\newblock \bibinfo{journal}{{\it Commun. ACM}} \bibinfo{volume}{21},
  \bibinfo{number}{7} (\bibinfo{year}{1978}), \bibinfo{pages}{558--565}.
\newblock


\bibitem[\protect\citeauthoryear{Lange, Ng, Toninho, and Yoshida}{Lange
  et~al\mbox{.}}{2017}]%
        {DBLP:conf/popl/LangeNTY17}
\bibfield{author}{\bibinfo{person}{Julien Lange}, \bibinfo{person}{Nicholas
  Ng}, \bibinfo{person}{Bernardo Toninho}, {and} \bibinfo{person}{Nobuko
  Yoshida}.} \bibinfo{year}{2017}\natexlab{}.
\newblock \showarticletitle{Fencing off go: liveness and safety for
  channel-based programming}. In \bibinfo{booktitle}{{\em Proc.\ of POPL'17}}.
  \bibinfo{publisher}{{ACM}}, \bibinfo{pages}{748--761}.
\newblock


\bibitem[\protect\citeauthoryear{Lange, Ng, Toninho, and Yoshida}{Lange
  et~al\mbox{.}}{2018}]%
        {Lange:2018:SVF:3180155.3180157}
\bibfield{author}{\bibinfo{person}{Julien Lange}, \bibinfo{person}{Nicholas
  Ng}, \bibinfo{person}{Bernardo Toninho}, {and} \bibinfo{person}{Nobuko
  Yoshida}.} \bibinfo{year}{2018}\natexlab{}.
\newblock \showarticletitle{A Static Verification Framework for Message Passing
  in Go Using Behavioural Types}. In \bibinfo{booktitle}{{\em Proc.\ of ICSE
  '18}}. \bibinfo{publisher}{ACM}, \bibinfo{pages}{1137--1148}.
\newblock


\bibitem[\protect\citeauthoryear{Mattern}{Mattern}{1989}]%
        {Mattern89virtualtime}
\bibfield{author}{\bibinfo{person}{Friedemann Mattern}.}
  \bibinfo{year}{1989}\natexlab{}.
\newblock \showarticletitle{Virtual Time and Global States of Distributed
  Systems}. In \bibinfo{booktitle}{{\em Parallel and Distributed Algorithms}}.
  \bibinfo{publisher}{North-Holland}, \bibinfo{pages}{215--226}.
\newblock


\bibitem[\protect\citeauthoryear{Netzer}{Netzer}{1994}]%
        {DBLP:conf/lcpc/Netzer93}
\bibfield{author}{\bibinfo{person}{Robert H.~B. Netzer}.}
  \bibinfo{year}{1994}\natexlab{}.
\newblock \showarticletitle{Trace Size vs. Parallelism in Trace-and-Replay
  Debugging of Shared-Memory Programs}. In \bibinfo{booktitle}{{\em Languages
  and Compilers for Parallel Computing, 6th International Workshop, Portland,
  Oregon, USA, August 12-14, 1993, Proceedings}} {\em
  (\bibinfo{series}{LNCS})}, Vol.~\bibinfo{volume}{768}.
  \bibinfo{publisher}{Springer}, \bibinfo{pages}{617--632}.
\newblock


\bibitem[\protect\citeauthoryear{Netzer and Miller}{Netzer and Miller}{1992}]%
        {DBLP:conf/sc/NetzerM92}
\bibfield{author}{\bibinfo{person}{Robert H.~B. Netzer} {and}
  \bibinfo{person}{Barton~P. Miller}.} \bibinfo{year}{1992}\natexlab{}.
\newblock \showarticletitle{Optimal Tracing and Replay for Debugging
  Message-Passing Parallel Programs}. In \bibinfo{booktitle}{{\em Proceedings
  Supercomputing '92, Minneapolis, MN, USA, November 16-20, 1992}}.
  \bibinfo{publisher}{{IEEE} Computer Society}, \bibinfo{pages}{502--511}.
\newblock


\bibitem[\protect\citeauthoryear{Ng and Yoshida}{Ng and Yoshida}{2016}]%
        {DBLP:conf/cc/NgY16}
\bibfield{author}{\bibinfo{person}{Nicholas Ng} {and} \bibinfo{person}{Nobuko
  Yoshida}.} \bibinfo{year}{2016}\natexlab{}.
\newblock \showarticletitle{Static deadlock detection for concurrent go by
  global session graph synthesis}. In \bibinfo{booktitle}{{\em Proc.\ of
  CC'16}}. \bibinfo{publisher}{{ACM}}, \bibinfo{pages}{174--184}.
\newblock


\bibitem[\protect\citeauthoryear{Reppy}{Reppy}{1999}]%
        {Reppy:1999:CPM:317040}
\bibfield{author}{\bibinfo{person}{John~H. Reppy}.}
  \bibinfo{year}{1999}\natexlab{}.
\newblock \bibinfo{booktitle}{{\em Concurrent Programming in {ML}}}.
\newblock \bibinfo{publisher}{Cambridge University Press},
  \bibinfo{address}{New York, NY, USA}.
\newblock
\showISBNx{0-521-48089-2}


\bibitem[\protect\citeauthoryear{Ronsse, Christiaens, and De~Bosschere}{Ronsse
  et~al\mbox{.}}{2003}]%
        {Ronsse:2003:DSM:860016.860024}
\bibfield{author}{\bibinfo{person}{Michiel Ronsse}, \bibinfo{person}{Mark
  Christiaens}, {and} \bibinfo{person}{Koen De~Bosschere}.}
  \bibinfo{year}{2003}\natexlab{}.
\newblock \showarticletitle{Debugging Shared Memory Parallel Programs Using
  Record/Replay}.
\newblock \bibinfo{journal}{{\em Future Gener. Comput. Syst.\/}}
  \bibinfo{volume}{19}, \bibinfo{number}{5} (\bibinfo{date}{July}
  \bibinfo{year}{2003}), \bibinfo{pages}{679--687}.
\newblock
\showISSN{0167-739X}


\bibitem[\protect\citeauthoryear{Ronsse and De~Bosschere}{Ronsse and
  De~Bosschere}{1999}]%
        {Ronsse:1999:RFI:312203.312214}
\bibfield{author}{\bibinfo{person}{Michiel Ronsse} {and} \bibinfo{person}{Koen
  De~Bosschere}.} \bibinfo{year}{1999}\natexlab{}.
\newblock \showarticletitle{RecPlay: A Fully Integrated Practical Record/Replay
  System}.
\newblock \bibinfo{journal}{{\em ACM Trans. Comput. Syst.\/}}
  \bibinfo{volume}{17}, \bibinfo{number}{2} (\bibinfo{date}{May}
  \bibinfo{year}{1999}), \bibinfo{pages}{133--152}.
\newblock
\showISSN{0734-2071}


\bibitem[\protect\citeauthoryear{Sen, Rosu, and Agha}{Sen
  et~al\mbox{.}}{2006}]%
        {DBLP:journals/sttt/SenRA06}
\bibfield{author}{\bibinfo{person}{Koushik Sen}, \bibinfo{person}{Grigore
  Rosu}, {and} \bibinfo{person}{Gul Agha}.} \bibinfo{year}{2006}\natexlab{}.
\newblock \showarticletitle{Online efficient predictive safety analysis of
  multithreaded programs}.
\newblock \bibinfo{journal}{{\em {STTT}\/}} \bibinfo{volume}{8},
  \bibinfo{number}{3} (\bibinfo{year}{2006}), \bibinfo{pages}{248--260}.
\newblock


\bibitem[\protect\citeauthoryear{Smaragdakis, Evans, Sadowski, Yi, and
  Flanagan}{Smaragdakis et~al\mbox{.}}{2012}]%
        {Smaragdakis:2012:SPR:2103621.2103702}
\bibfield{author}{\bibinfo{person}{Yannis Smaragdakis}, \bibinfo{person}{Jacob
  Evans}, \bibinfo{person}{Caitlin Sadowski}, \bibinfo{person}{Jaeheon Yi},
  {and} \bibinfo{person}{Cormac Flanagan}.} \bibinfo{year}{2012}\natexlab{}.
\newblock \showarticletitle{Sound Predictive Race Detection in Polynomial
  Time}.
\newblock \bibinfo{journal}{{\em SIGPLAN Not.\/}} \bibinfo{volume}{47},
  \bibinfo{number}{1} (\bibinfo{date}{Jan.} \bibinfo{year}{2012}),
  \bibinfo{pages}{387--400}.
\newblock
\showISSN{0362-1340}


\bibitem[\protect\citeauthoryear{Stadtm{\"{u}}ller, Sulzmann, and
  Thiemann}{Stadtm{\"{u}}ller et~al\mbox{.}}{2016}]%
        {DBLP:conf/aplas/StadtmullerST16}
\bibfield{author}{\bibinfo{person}{Kai Stadtm{\"{u}}ller},
  \bibinfo{person}{Martin Sulzmann}, {and} \bibinfo{person}{Peter Thiemann}.}
  \bibinfo{year}{2016}\natexlab{}.
\newblock \showarticletitle{Static Trace-Based Deadlock Analysis for
  Synchronous Mini-Go}. In \bibinfo{booktitle}{{\em Proc.\ of APLAS'16}} {\em
  (\bibinfo{series}{LNCS})}, Vol.~\bibinfo{volume}{10017}.
  \bibinfo{pages}{116--136}.
\newblock


\bibitem[\protect\citeauthoryear{Sulzmann and Stadtm{\"{u}}ller}{Sulzmann and
  Stadtm{\"{u}}ller}{2017}]%
        {DBLP:conf/hvc/SulzmannS17}
\bibfield{author}{\bibinfo{person}{Martin Sulzmann} {and} \bibinfo{person}{Kai
  Stadtm{\"{u}}ller}.} \bibinfo{year}{2017}\natexlab{}.
\newblock \showarticletitle{Trace-Based Run-Time Analysis of Message-Passing Go
  Programs}. In \bibinfo{booktitle}{{\em Proc.\ of HVC'17}} {\em
  (\bibinfo{series}{LNCS})}, Vol.~\bibinfo{volume}{10629}.
  \bibinfo{publisher}{Springer}, \bibinfo{pages}{83--98}.
\newblock


\bibitem[\protect\citeauthoryear{Tasharofi, Pradel, Lin, and Johnson}{Tasharofi
  et~al\mbox{.}}{2013}]%
        {Tasharofi:2013:BCA:3107656.3107674}
\bibfield{author}{\bibinfo{person}{Samira Tasharofi}, \bibinfo{person}{Michael
  Pradel}, \bibinfo{person}{Yu Lin}, {and} \bibinfo{person}{Ralph Johnson}.}
  \bibinfo{year}{2013}\natexlab{}.
\newblock \showarticletitle{Bita: Coverage-guided, Automatic Testing of Actor
  Programs}. In \bibinfo{booktitle}{{\em Proceedings of the 28th IEEE/ACM
  International Conference on Automated Software Engineering}} {\em
  (\bibinfo{series}{ASE'13})}. \bibinfo{publisher}{IEEE Press},
  \bibinfo{address}{Piscataway, NJ, USA}, \bibinfo{pages}{114--124}.
\newblock
\showISBNx{978-1-4799-0215-6}


\bibitem[\protect\citeauthoryear{Vo}{Vo}{2011}]%
        {Vo:2011:SFD:2231450}
\bibfield{author}{\bibinfo{person}{Anh Vo}.} \bibinfo{year}{2011}\natexlab{}.
\newblock {\em \bibinfo{title}{Scalable Formal Dynamic Verification of Mpi
  Programs Through Distributed Causality Tracking}}.
\newblock \bibinfo{thesistype}{Ph.D. Dissertation}. \bibinfo{school}{University
  of Utah}.
\newblock Advisor(s) Gopalakrishnan, Ganesh.
\newblock
\showISBNx{978-1-124-63005-2}
\newblock
\shownote{AAI3454168.}


\bibitem[\protect\citeauthoryear{Vo, Gopalakrishnan, Kirby, de~Supinski,
  Schulz, and Bronevetsky}{Vo et~al\mbox{.}}{2011}]%
        {DBLP:conf/IEEEpact/VoGKSSB11}
\bibfield{author}{\bibinfo{person}{Anh Vo}, \bibinfo{person}{Ganesh
  Gopalakrishnan}, \bibinfo{person}{Robert~M. Kirby},
  \bibinfo{person}{Bronis~R. de Supinski}, \bibinfo{person}{Martin Schulz},
  {and} \bibinfo{person}{Greg Bronevetsky}.} \bibinfo{year}{2011}\natexlab{}.
\newblock \showarticletitle{Large Scale Verification of {MPI} Programs Using
  Lamport Clocks with Lazy Update}. In \bibinfo{booktitle}{{\em Proc.\ of
  PACT'11}}. \bibinfo{publisher}{{IEEE} Computer Society},
  \bibinfo{pages}{330--339}.
\newblock


\bibitem[\protect\citeauthoryear{Wang, Kundu, Ganai, and Gupta}{Wang
  et~al\mbox{.}}{2009}]%
        {Wang:2009:SPA:1693345.1693367}
\bibfield{author}{\bibinfo{person}{Chao Wang}, \bibinfo{person}{Sudipta Kundu},
  \bibinfo{person}{Malay Ganai}, {and} \bibinfo{person}{Aarti Gupta}.}
  \bibinfo{year}{2009}\natexlab{}.
\newblock \showarticletitle{Symbolic Predictive Analysis for Concurrent
  Programs}. In \bibinfo{booktitle}{{\em Proc.\ of FM'09}}.
  \bibinfo{publisher}{Springer-Verlag}, \bibinfo{address}{Berlin, Heidelberg},
  \bibinfo{pages}{256--272}.
\newblock
\showISBNx{978-3-642-05088-6}
\showURL{%
\url{http://dx.doi.org/10.1007/978-3-642-05089-3_17}}


\end{thebibliography}

\appendix

\section{Trace Replay for Buffered Channels}
\label{app-non-backtrack-replay}

We consider the issue of identifying alternative schedules.
Example~\ref{ex:buffer-one} suggests that a possible candidate
for an alternative schedule arises if rule \rlabel{Send}
is applicable on two different thread-local traces that share
the same channel.
We ask the question, is this a sufficient condition to identify some alternative schedule?
The answer is no as shown by the following example.

\begin{example}
  \label{ex:trace-replay-stuck}
  Consider
\bda{lcl}
    [ x \assign \MAKECHAN{2},
      \\ \GO\ [\SEND{x}{1}],   && (1)
      \\ \GO\ [\SEND{x}{1}],   && (2)
      \\ \RCV{x},              && (3)
      \\ \RCV{x} ]             && (4)

\eda
where we assume a program run where send at location (1) is executed first,
then the receive at location (3) proceeds, followed by the send at location (2).
Finally, the receive at location (4) is executed.
Here is the resulting run-time trace.
\bda{l}
    [\thread{1}{[\signalTrace{2}, \signalTrace{3},
        \pre{\rcv{x}}, \postRcv{2}{1}{x},
        \pre{\rcv{x}}, \postRcv{3}{1}{x}]},
    \\ \thread{2}{[\waitTrace{2}, \pre{\snd{x}}, \postSnd{2}{1}{x}]},
    \\ \thread{3}{[\waitTrace{3}, \pre{\snd{x}}, \postSnd{3}{1}{x}]} ]
\eda

Trace replay rule \rlabel{Send} could be applied on the event
resulting from location (1) or on the event resulting from location (2).

Suppose, we first process the event from location (2).
Next, we process the event from location (3).
This leads to the buffer
$$
[\postSnd{3}{1}{x}^{[0,0,2]}, \postSnd{2}{1}{x}^{[0,2,0]}]
$$
At this point we are stuck!
We cannot proceed, i.e.~apply rule \rlabel{Receive}, as the send events are buffered in the wrong order.
If we  switch the order and first process the event from location (2),
this then leads to successful trace replay run.
We conclude that there are no alternative schedules for this program run.
\end{example}

Can we somehow discard the first choice?
The receive connected to event $\postSnd{3}{1}{x}$
from location (3) happens later compared to the receive
connected to event $\postSnd{2}{1}{x}$ from location (2).
This can easily be derived from the trace position of each receive.
During trace replay we ensure that if there are two competing
send events, i.e.~their receives are from the same thread,
we favor the send with the earlier receive.

The trace position of the corresponding send can be computed in $O(m * m)$.
There is a unique connection among send-receive pairs.
For each send event $\postSnd{i}{j}{x}$
we can identify its corresponding receive event $\postRcv{i}{j}{x}$
via the sender's thread id~$i$ and the program counter~$j$.
There are $O(m)$ candidates and for each candidate there are $O(m)$ potential
receive partners to consider.
We annotate each send with the
thread id and the trace position in the trace of that thread
of its receiver partner.
Suppose $\postRcv{i}{j}{x}$ is found in trace with the id~$l$ at position $k$.
Then, we write $\postSnd{i}{j}{x}^{(l,k)}$.
We assume recorded send events carry the additional annotation.

\begin{definition}[Replay Strategy]
\label{def:replay-strategy}  
Whenever we apply rule \rlabel{Send} on $\postSnd{i}{j}{x}^{(l,k)}$
we ensure that there is no competing send $\postSnd{i'}{j'}{x}^{(l',k')}$
where $l=l$ and $k' < k'$.
\end{definition}

\begin{proposition}
  Any vector clock annotation derived via some exhaustive trace replay run
  can also be derived via the replay strategy specified in
  Definition~\ref{def:replay-strategy}.
\end{proposition}
\begin{proof}
  Consider some exhaustive trace replay run.
  Events for a particular thread are recorded  in the order as they are executed.
  Consider a specific trace where there are two consecutive receive events on the same channel.
  The earlier receive event will be processed first.
  As receives events are uniquely connected to their corresponding send events,
  this implies that the send event of the earlier receive must be processed before
  the send event of the later receive.
  This immediately holds for unbuffered channels. For buffered channels, the statement holds
  because buffered messages are queued.
  We conclude that the condition imposed by the replay strategy in  Definition~\ref{def:replay-strategy}
  do not affect any exhaustive trace replay run.
\end{proof}  

We come back to the above question.
Under the trace replay strategy in Definition~\ref{def:replay-strategy},
if rule \rlabel{Send} is applicable on two different thread-local traces that share
the same channel, does this imply that there is some alternative schedule?
The answer is still no as the next example shows.

\begin{example}
\label{ex:stuck-and-deadlock}  
  Consider
  \bda{lcl}
      [ x \assign \MAKECHAN{1},
        \\ y \assign \MAKECHAN{1},
        \\ \GO\ [\SEND{x}{1},        && (1)
          \\ \ \ \ \ \ \ \ \ \ \ \ \ \ \SEND{y}{1} ],  && (2)
        \\ \GO [\RCV{x} ],       && (3)
        \\ \GO [\SEND{x}{1}],    && (4)
        \\ \RCV{y},              && (5)
        \\ \RCV{x} ]             && (6)
      \eda
      We assume communications are as follows: (1) with (3),
      (2) with (5) and (4) with (6).
      Here is the resulting trace where assume that signal/trace events have already been processed.
      Main thread has the id~1. Helper threads are given increasing numbers from top to bottom starting with the id~2.      
  \bda{l}
      [\thread{1}{[\pre{\rcv{y}}, \postRcv{2}{2}{y}, \pre{\rcv{x}}, \postRcv{4}{1}{x} ]},
      \\ \thread{2}{[\pre{\snd{x}}, \postSnd{2}{1}{x}, \pre{\snd{y}}, \postSnd{2}{2}{y}]},
      \\ \thread{3}{[\pre{\rcv{x}}, \postRcv{2}{1}{x}]},
      \\ \thread{4}{[\pre{\snd{x}}, \postSnd{4}{1}{x}]} ]
      \eda

      We are free to choose any order among trace replay rules as long as we respect Definition~\ref{def:replay-strategy}.
      We could either apply rule \rlabel{Send} on thread~2 or on thread~4.
      Both choices are possible as the corresponding receive events are in distinct threads.
      We choose to apply \rlabel{Send} on thread~4 first.
      This leads to the following trace.
  \bda{l}
      [\thread{1}{[\pre{\rcv{y}}, \postRcv{2}{2}{y}, \pre{\rcv{x}}, \postRcv{4}{1}{x} ]},
      \\ \thread{2}{[\pre{\snd{x}}, \postSnd{2}{1}{x}, \pre{\snd{y}}, \postSnd{2}{2}{y}]},
      \\ \thread{3}{[\pre{\rcv{x}}, \postRcv{2}{1}{x}]},
      \\ \thread{4}{[]} ]
      \eda
      The buffer of channel $y$ is empty and $x$'s buffer consists of $[\postSnd{4}{1}{x}]$.
      For brevity, we ignore vector clock annotations.
      At this point we are stuck. No further trace replay rules are applicable and
      we still have unprocessed post events.
      If we apply rule \rlabel{Send} first
      on thread~2, this leads to a successful (exhaustive) trace replay run.
      For brevity we omit the details.
\end{example}

We conclude.
Application of rule \rlabel{Send} on two different thread-local traces that share
the same channel is a necessary but not sufficient criteria to identify some alternative schedule.
For trace replay (in case of buffered channels) this means that we might need to backtrack.
At this point, we do not know of a strategy that helps us to eliminate all choices that lead to
a stuck state. Recall, stuck means that no further trace replay rules are applicable
but the trace contains some unprocessed post events.

We might be tempted to interpret a stuck state as a deadlocking situation.
However, this is not necessarily the case as the following example shows.

\begin{example}
Consider the following program where we use Go-style syntax.  
  \begin{verbatim}
x := make(chan 1)
y := make(chan 1)
go { x <- 1        (1)
     y <- 1        (2)
   }
go { z:= <-x       (3)
     if z == 2 {
       <-x         (4)
       x <- 1      (5)
     }
   }
go { x <- 2 }      (6)
<-y                (7)
<-x                (8)
\end{verbatim}
The program is deadlock-free as the following observation shows.
Suppose (1) executes first. We cannot proceed with (6) as the buffer space is occupied.
Then, (3) executes. As the value received is 1,
(4) and (5) do not apply for this program run.
We find that (3) executes, then (7).
Finally, (4) executes followed by (8).

Suppose (6) executes first. Then, (1) is blocked.
We execute (3). The value received is 2, hence, (4) and (5) become relevant.
The buffer space is empty. Hence, we execute (1) followed by (4) and (5).
We execute (2) followed by (7). Finally, we execute (8).

We consider a program run where (1) executes first.
The resulting trace is as follows.
As usual, the main thread has id~1. The other threads
are given increasing numbers from top to bottom starting with id~2.
We ignore the program counter (by using the don't care value $\_$)
as the concrete values do not matter here.
\bda{ll}
    [\thread{1}{[\pre{\rcv{y}}, \postRcv{2}{\_}{y},
          \pre{\rcv{x}}, \postRcv{4}{\_}{x}]},
      \\  \thread{2}{[\pre{\snd{x}}, \postSnd{2}{\_}{x},
          \pre{\snd{y}}, \postSnd{2}{\_}{y}]},
      \\ \thread{3}{[\pre{\rcv{x}}, \postRcv{2}{\_}{x}]},
      \\ \thread{4}{[\pre{\snd{x}}, \postSnd{4}{\_}{x}]}
    ]
\eda
    
We consider trace replay.
We could either apply rule \rlabel{Send} on thread 2 or on thread 4.
Suppose we choose thread 4 which leads to the following trace.
\bda{ll}
    [\thread{1}{[\pre{\rcv{y}}, \postRcv{2}{\_}{y},
          \pre{\rcv{x}}, \postRcv{4}{\_}{x}]},
      \\  \thread{2}{[\pre{\snd{x}}, \postSnd{2}{\_}{x},
          \pre{\snd{y}}, \postSnd{2}{\_}{y}]},
      \\ \thread{3}{[\pre{\rcv{x}}, \postRcv{2}{\_}{x}]},
      \\ \thread{4}{[]}
    ]
\eda
The buffer of $y$ is empty and $x$'s buffer consists of
$[\postSnd{4}{\_}{x}]$.
At this point we are stuck.
We can not apply rule \rlabel{Send} on thread 2 because the buffer space is occupied.
We can not apply rule \rlabel{Receive} on thread 3 because
the element in $x$'s buffer is not the proper partner.

As observed above, the program is deadlock-free.
This shows that being stuck is not a sufficient condition to conclude
that that there is a deadlock.
\end{example}

Here comes some good news. For certain stuck situations we can identify a deadlocking
situation.

\begin{definition}
  We consider a stuck state during trace replay.
  We say that this state is \emph{completely stuck} iff
  for any post receive event in leading position in some trace,
  the buffer space of the respective channel is empty.
\end{definition}
Note that completely stuck implies if there is some post send event
in leading position, the buffer space of the respective channel is fully occupied.

\begin{proposition}
Completely stuck implies that there is a program run that leads to a deadlock.
\end{proposition}
\begin{proof}
  Trace replay rules mimic operational semantic rules where
  we assume a fixed connection among send-receive pairs.
  So, any sequence of trace replay rule steps corresponds to an actual program run.
  Consider the completely stuck state.
  No further rules are applicable.
  There are no alternative matchings among send-receive pairs available to leave
  this (stuck) state. Hence, we have reached a deadlocking situation.
\end{proof}

\begin{example}
Consider the following program.
\begin{verbatim}
x := make(chan 1)
y := make(chan 1)
go { x <- 1         (1)
     y <- 1         (2)
     y <- 1         (3)
   }
go { x <- 1         (4)
   }
go { <-y            (5)
     <-x            (6)
   }
<-y                 (7)
<-x                 (8)
\end{verbatim}
We consider a program run where
(1) communicates with (6),
(2) with (5),
(3) with (7) and (4) with (8).
The resulting trace is as follows.
\bda{ll}
    [
      \thread{1}{[\pre{\rcv{y}}, \postRcv{2}{\_}{y},
                  \pre{\rcv{x}}, \postRcv{3}{\_}{x}]},
\\      \thread{2}{[\pre{\snd{x}}, \postSnd{2}{\_}{x},
                  \pre{\snd{y}}, \postSnd{2}{\_}{y},
                  \pre{\snd{y}}, \postSnd{2}{\_}{y} ]},
\\      \thread{3}{[\pre{\snd{x}}, \postSnd{3}{\_}{x}]},
\\      \thread{4}{[\pre{\rcv{y}}, \postRcv{2}{\_}{y},
                    \pre{\rcv{x}}, \postRcv{2}{\_}{x}]}
    ]
\eda

Now, consider trace replay where we apply rule \rlabel{Send} on thread 3.
This leads to
\bda{ll}
    [
      \thread{1}{[\pre{\rcv{y}}, \postRcv{2}{\_}{y},
                  \pre{\rcv{x}}, \postRcv{3}{\_}{x}]},
\\      \thread{2}{[\pre{\snd{x}}, \postSnd{2}{\_}{x},
                  \pre{\snd{y}}, \postSnd{2}{\_}{y},
                  \pre{\snd{y}}, \postSnd{2}{\_}{y} ]},
\\      \thread{3}{[]},
\\      \thread{4}{[\pre{\rcv{y}}, \postRcv{2}{\_}{y},
                    \pre{\rcv{x}}, \postRcv{2}{\_}{x}]}
    ]
\eda
where $y$'s buffer space is empty and $x$'s buffer space
consists of $[\postSnd{3}{\_}{x}]$.
We are completely stuck.
Indeed, the program run implied by trace replay
results in a deadlock.
\end{example}

\section{Alternative Tracing Schemes}
\label{app:alternatives}

Suppose we only trace the sender's thread id (but not the program counter).
We write $\post{\link{i}{\rcv{x}}}$ to denote a committed receive operation
where the value is obtained from thread~$i$.
For committed sends, we simply write $\post{\snd{x}}$.

\begin{example}
Consider
\bda{lcl}
    [ x \assign \SYNCMAKECHAN, && (1) \\
      \GO\ [\RCV{x}], && (2) \\
      \GO\  [\RCV{x}], && (3) \\
      \SEND{x}{2}, \SEND{x}{3} ]
\eda
We assume a specific program run where the receive in thread~2
obtains the value~2 and the other receive obtains the value~3.
Here is the resulting trace.

\bda{ll}
    [\thread{1}{[\signalTrace{2}, \signalTrace{3}, \pre{\snd{x}}, \post{\snd{x}}, \pre{\snd{x}}, \post{\snd{x}}]},
   \\ \thread{2}{[\waitTrace{2}, \pre{\rcv{x}}, \post{\link{1}{\rcv{x}}}]},
   \\ \thread{3}{[\waitTrace{3}, \pre{\rcv{x}}, \post{\link{1}{\rcv{x}}}]}
   \eda
   After two consecutive applications of rule \rlabel{Signal/Wait}, we reach
   the trace.
\bda{ll}
    [\thread{1}{[\pre{\snd{x}}, \post{\snd{x}}, \pre{\snd{x}}, \post{\snd{x}}]},
   \\ \thread{2}{[ \pre{\rcv{x}}, \post{\link{1}{\rcv{x}}}]},
   \\ \thread{3}{[ \pre{\rcv{x}}, \post{\link{1}{\rcv{x}}}]}
   \eda
   There are two possible ways to synchronize.
   Either involving thread 1 and 2 or thread 1 and 3.
   The second option leads to a non-predictable program run
   as we assume the first send is received by thread 2.
\end{example}

In the extreme case, we could ignore any connections between send-receive pairs.
This yields a tracing scheme with very low run-time overhead.
The consequence is that the analysis is no longer predictive and false positives may arise.
Below we consider an example where the analysis suggests there is a deadlock where there is
in reality none.

\begin{example}
  Consider following the program in Go-style syntax.
\begin{verbatim}
x := make(chan 0)              (Thread 1)
z := make(chan 0)
go { x <- 1 }                  (Thread 2)
go { x <- 2                    (Thread 3)
       z <- 1 }
y := <-x
if y == 2 {
    <-z
   <-x
} else {
   <-x
   <-z
}
\end{verbatim}  
We consider a program run for which $y$ holds the value 2.
That means, thread 3 transmits the value 2 first.

Here is the (abbreviated) list of run-time traces. We only record post events.
Recall that no links between send-receive pairs are recorded.

\bda{ll}
    [\thread{1}{[\rcv{x}, \rcv{z}, \rcv{x}]},
      \\ \thread{2}{[\snd{x}]},
      \\ \thread{3}{[\snd{x}, \snd{z}]} ]
\eda
Trace replay may get stuck here. First example, if we first process thread 2 and then thread 2.
\bda{ll}
    [\thread{1}{[\rcv{z}, \rcv{x}]},
      \\ \thread{2}{[]},
      \\ \thread{3}{[\snd{x}, \snd{z}]} ]
\eda
We might be tempted to interpret being stuck as the program contains a deadlock.
This is not the case!
The program is deadlock-free.
The issue is that the trace replay chosen does not correspond to any actual program run.
\end{example}

\section{Comparison to \cite{DBLP:conf/hvc/SulzmannS17}}
\label{sec:comparison-dg}

\subsection{Dependency Graph Construction}

We repeat the details of the dependency graph construction.
We follow the notation used in~\cite{DBLP:conf/hvc/SulzmannS17}.
Each pre is followed by a post event.
This allows for a more uniform construction of the dependency graph
and the analysis that is carried out on the graph.
Adding a dummy post event to each dangling pre event can be achieved
via a simple scan through the list of traces.
For example,

$$[\thread{1}{T_1 \pp [\pre{[\rcv{x}]}]}, \dots, \thread{n}{T_n}]$$

is transformed into

$$[\thread{1}{T_1 \pp [\pre{[\rcv{x}]}, \post{\thread{n_1}{\post{\rcv{x}}}}]}, \dots, \thread{n}{T_n}, \thread{n+1}{[\pre{[\snd{x}]}, \post{\snd{x}}]}]$$
We assume that this transformation is applied until there are no dangling pre events left.

We further assume that all post events are annotated with the program location
of the preceding send/receive operation.
For example, $\post{\ploc{\snd{x}}{4}}$ denotes a post event
connected to a send via channel $x$ at program location $4$.
For dummy post events added, we assume some dummy program locations.

\begin{definition}[Construction of Dependency Graph]
\label{def:dg}
The dependency graph is a directed graph $G=(N,V)$ where $N$ denotes
the set of nodes and $E$ denotes
the set of (directed) edges, represented as pairs of nodes.
Nodes correspond to post events (annotated with program locations).
We obtain the dependency graph from the transformed list $[\thread{1}{T_1}, \dots, \thread{n}{T_n}]$
of traces as follows.
We write $\delta$ to denote $\snd{x}$, $\rcv{x}$, $\close{x}$ and $\select$.
This characterizes all possibles shapes of post events (ignoring program locations).

\bda{lcl}
N & = & \{ \ploc{\delta}{l} \mid \exists i. T_i = [\dots,\post{\ploc{\delta}{l}},\dots] \}
\\ & \cup & \{ \close{x} \mid \exists i,l. T_i = [\dots,\post{\ploc{\close{x}}{l}},\dots] \}
 \\
 \\
 E & = & \{  (\ploc{\delta}{l}, \ploc{\delta'}{l'}) \mid
 \\ & & \exists i. T_i = [\dots,\pre{\dots},\post{\ploc{\delta}{l}},\pre{\dots},\post{\ploc{\delta'}{l'}},\dots] \} \ (1)
 \\ & \cup &
 \{ (\ploc{\snd{x}}{l}, \ploc{\rcv{x}}{l'}) \mid
 \\ && \exists i,j. T_i = [\dots,\pre{\dots},\post{\ploc{\snd{x}}{l}},\dots] 
 \\ & & \ \ \ \ \ \ \ \ \ \ \ \ \ \ \ \ \ \ \ \ \ \ \ \ \ \ \  T_j = [\dots,\pre{\dots},\post{\ploc{\thread{i}{\rcv{x}}}{l}},\dots] \} \ (2)
 \\ & \cup &
 \{ (\close{x}, \ploc{\rcv{x}}{l}) \mid
  \\ && \exists i. T_i = [\dots,\pre{\dots},\post{\ploc{\thread{0}{\rcv{x}}}{l}},\dots] \} \ (3)
\eda
\end{definition}  
Each post event is turned into a node. Recall that all dangling pre events in the initial list of traces
are provided with a dummy post event.
For all close operations on some channel $x$, we generate a node that does not reference
the program location. The reason will be explained shortly.
For each trace, pre/post events take place in sequence.
Hence, there is an edge from each to node the following (as in the program text) node.
See (1).
For each send/receive synchronization we find another edge. See (2).
We consider the last case (3).

A receive can also synchronize due to a closed channel.
This is easy to spot as the thread id number attached to the post event is the dummy value~$0$.
To identify the responsible close operation we would require a replay.
We avoid this extra cost and overapproximate by simply drawing
an edge from node $\close{x}$ to to $\ploc{\rcv{x}}{l}$.
Node $\close{x}$ is the representative for one of the close operations in the program text.

To check if one event happens-before another event we seek for
a path from one event to the other.
Two events are concurrent if neither happens-before the other.
To check for alternative communications, we check for
matching nodes that are concurrent to each other.
By matching we mean that one of the nodes is a send and the other is a receive
over the same channel.

\subsection{Limitations}

For the detection of alternative communications,
the dependency graph is either too optimistic or pessimistic,
depending if we ignore or include inter-thread connections.
For intra-thread dependencies, the dependency graph turns out to be too optimistic.
We explain the above points via some examples.

Consider the following program. 
\bda{lcl}
    [ x \assign \SYNCMAKECHAN,      && (1)
      \\ \GO\ [\SEND{x}{1}],        && (2)
      \\ \GO\ [\RCV{x}],             && (3)
      \\ \GO\ [\RCV{x}, \SEND{x}{1}]] && (4)

\eda
In the above, we omit explicit program locations.
We assume that the send operation in thread 2 is connected to location 2 and so forth.
In case of thread 4, the receive operation is connected to location 4 and
the send operation connected to location 4'.

We consider a program run where thread 2 synchronizes with thread 4 and
then thread 3 synchronizes with thread 4.
This leads to the following trace.
As for the earlier example, we ignore $\signalTrace{\cdot}$ and $\waitTrace{\cdot}$ events.

\bda{l}
    [\thread{1}{[]},
      \\ \thread{2}{[\pre{\snd{x}}, \post{\ploc{\snd{x}}{2}}]},
      \\ \thread{3}{[\pre{\rcv{x}}, \post{\thread{4}{\ploc{\rcv{x}}{3}}}]},
      \\ \thread{4}{[\pre{\rcv{x}}, \post{\thread{2}{\ploc{\rcv{x}}{4}}}, \pre{\snd{x}}, \post{\ploc{\snd{x}}{4'}}]}
    ]
\eda
We derive the following dependency graph.

\bda{c}
\includegraphics[scale=0.5]{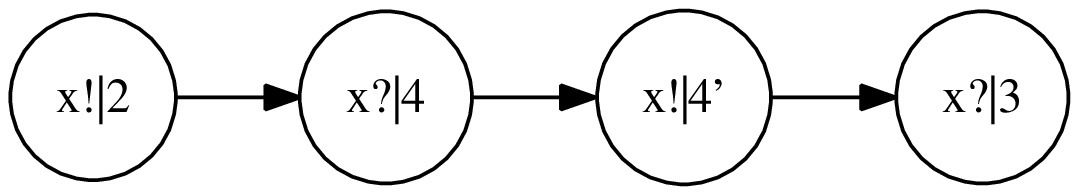}
\eda
It seems that there are no alternative communications.
Matching events $\snd{x}|2$ and $\rcv{x}|3$
are not concurrent because $\rcv{x}|3$ is reachable from $\snd{x}|2$.
On the other hand, based on the trace (replay), it follows immediately
that $\snd{x}|2$ and $\rcv{x}|3$ represent some alternative communication.
Hence, the graph representation is too pessimistic here.

Could we adjust the dependency graph by ignoring inter-thread dependencies?
In terms of the graph construction, see Definition~\ref{def:dg},
we ignore edges
\bda{l}
 \{  (\ploc{\delta}{l}, \ploc{\delta'}{l'}) \mid
 \\ \exists i. L_i = [\dots,\pre{\dots},\post{\ploc{\delta}{l}},\pre{\dots},\post{\ploc{\delta'}{l'}},\dots] \}
\eda
For the above example, this removes the edge from $\rcv{x}|4$ to $\snd{x}|4'$
and then $\snd{x}|2$ and $\rcv{x}|3$ are concurrent to each other.
Unfortunately, the thus adjusted dependency graph may be overly optimistic.

Consider
\bda{lcl}
    [ x \assign \SYNCMAKECHAN,         && (1)
\\      \GO\ [\SEND{x}{1}, \SEND{x}{1}],  && (2)
\\      \GO\ [\RCV{x}],                   && (3)
\\      \GO\ [\RCV{x}, \SEND{x}{1}, \RCV{x}]] && (4)
\eda
As before, we omit explicit program locations.
In case of thread 4, the sequence of receive, send and receive operations
is (implicitly) labeled with locations 4, 4' and 4''.
In thread 2, the first send is labeled with location 2 and
the second second with location 2'.

We consider a program run where thread 2 synchronizes with thread 4.
Thread 4 synchronizes with thread 3 and finally thread 2 synchronizes with thread 4.
Below is the resulting trace and the dependency graph derived from the trace.

\bda{l}
    [\thread{1}{[]},
\\   \thread{2}{[\pre{\snd{x}}, \post{\ploc{\snd{x}}{2}}, \pre{\snd{x}}, \post{\ploc{\snd{x}}{2'}}]},
\\   \thread{3}{[\pre{\rcv{x}}, \post{\thread{4}{\ploc{\rcv{x}}{3}}}]},
\\   \thread{4}{[\pre{\rcv{x}}, \post{\thread{2}{\ploc{\rcv{x}}{4}}},
                 \pre{\snd{x}}, \post{\ploc{\snd{x}}{4'}},
                 \pre{\rcv{x}}, \post{\thread{2}{\ploc{\rcv{x}}{4''}}}]}
    ]

\eda

\bda{c}
\includegraphics[scale=0.5]{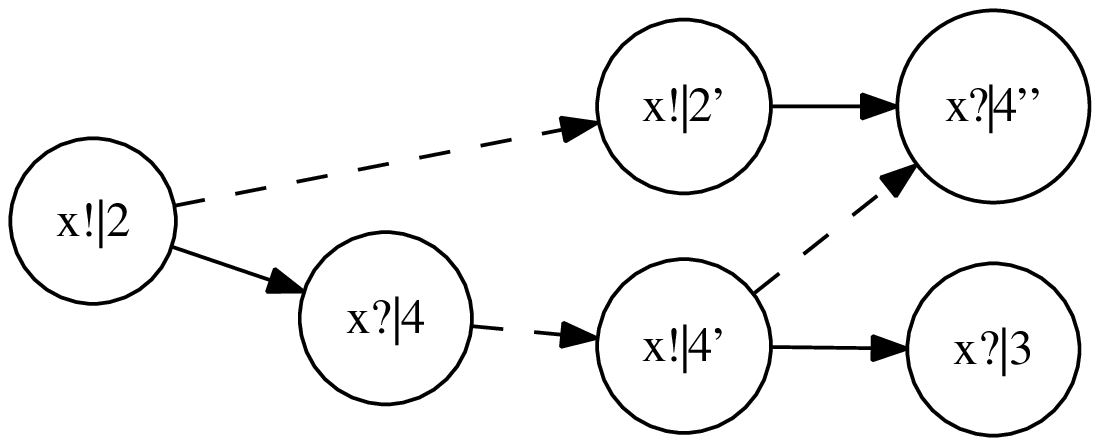}
\eda

Dashed edges represent inter-thread dependencies.
Assuming we ignore inter-thread dependencies,
matching events $\snd{x}|2$ and $\rcv{x}|3$ are concurrent
to each other as neither node can be reached from the other node.
This is correct and consistent with the trace replay method.
However, we also identify $\snd{x}|2$
and $\rcv{x}|4''$ as concurrent to each other.
This is too optimistic. There is no schedule where
$\snd{x}|2$
and $\rcv{x}|4''$ can be concurrent to each other.

Besides inter-thread dependencies, we may also encounter intra-thread dependencies.
Consider the following program.
\bda{lcl}
    [ x \assign \SYNCMAKECHAN,              && (1)
\\      \GO\ [\SEND{x}{1}],                   && (2)
\\      \GO\ [\RCV{x},                        && (3)
\\ \ \ \ \ \ \GO\ [\SEND{x}{1}],   && (4)
\\ \ \ \ \ \ \GO\ [\RCV(x)]]]   && (5)        
\eda
Within thread 3, we create thread 4 and 5.
Thread 4 and 5 will only become active once the receive operation in thread 3 synchronizes
with the send operation in thread 2.
Our trace replay method takes care of this intra-thread dependency
via $\signalTrace{\cdot}$ and $\waitTrace{\cdot}$ events.
In our construction of the dependency graph, we ignore such events
and therefore the dependency graph is too optimistic.
Based on the dependency graph, not shown here, we would derive (wrongly)
derive that receive operation in thread~5 is an alternative communication partner
for the send operation in thread~2.

Another limitation is that the tracing scheme in \cite{DBLP:conf/hvc/SulzmannS17}
only records the sender's thread id (and not the program counter as well).
As observed in Appendix~\ref{app:alternatives}.
This implies that the analysis is non-predictive.

\end{document}